\documentclass[11pt,a4paper]{article}
\pdfoutput=1

\usepackage{amssymb}
\usepackage{amsmath}
\usepackage{amsfonts}
\usepackage{bbm}
\usepackage{amsthm}
\usepackage{mathrsfs}
\usepackage{hyperref}
\usepackage{color}
\usepackage[margin=2.41cm]{geometry}
\usepackage[all,cmtip]{xy}
\usepackage[utf8]{inputenc}
\usepackage{graphicx}
\usepackage{varwidth}
\usepackage{comment}
\usepackage{enumitem}

\usepackage{mathtools}

\usepackage{upgreek}
\usepackage{rotating}

\usepackage{tikz}
\usetikzlibrary{shapes.geometric}

\usepackage{tikz-cd}
\usetikzlibrary{matrix,arrows,calc,decorations.pathmorphing,fit,shapes.geometric,decorations.pathreplacing,positioning}

\definecolor{darkred}{rgb}{0.8,0.1,0.1}
\hypersetup{
	colorlinks=true,         
	linkcolor=darkred,
	citecolor=blue,
}

\theoremstyle{plain}
\newtheorem{theo}{Theorem}[section]

\newtheorem{propo}[theo]{Proposition}

\theoremstyle{definition}
\newtheorem{defi}[theo]{Definition}

\newtheorem{exx}[theo]{Example}
\newenvironment{ex}
{\pushQED{\qed}\exx}
{\popQED\endexx}

\newtheorem{remm}[theo]{Remark}
\newenvironment{rem}
{\pushQED{\qed}\remm}
{\popQED\endremm}

\numberwithin{equation}{section}

\def\nn{\nonumber}

\def\bbK{\mathbb{K}}
\def\bbR{\mathbb{R}}
\def\bbC{\mathbb{C}}
\def\bbN{\mathbb{N}}
\def\bbZ{\mathbb{Z}}

\def\Hom{\mathrm{Hom}}
\def\hom{\underline{\mathrm{hom}}}

\def\Sym{\mathrm{Sym}}

\def\id{\mathrm{id}}

\def\dd{\mathrm{d}}
\def\vol{\mathrm{vol}}

\def\cc{\mathrm{c}}

\def\dR{\mathrm{dR}}

\def\1{I}

\def\op{\mathrm{op}}
\def\ev{\mathrm{ev}}

\def\pr{\mathrm{pr}}

\def\Ch{\mathbf{Ch}}

\def\dgCAlg{\mathbf{dgCAlg}}

\def\g{\mathfrak{g}}

\def\FFF{\mathfrak{F}}

\def\sk{\vspace{2mm}}

\makeatletter
\let\@fnsymbol\@alph
\makeatother

\title{%
Obstructions to the existence of M{\o}ller maps
}

\author{%
Marco Benini$^{1,2,a}$, Alastair Grant-Stuart$^{3,b}$, \vspace{1mm}\\
Giorgio Musante$^{1,c}$\ and\ Alexander Schenkel$^{3,d}$\vspace{4mm}\\
{\small ${}^1$ Dipartimento di Matematica, Dipartimento di Eccellenza 2023-27, Universit\`a di Genova,}\\
{\small Via Dodecaneso 35, 16146 Genova, Italy.}\vspace{2mm}\\
{\small ${}^2$ INFN, Sezione di Genova,}\\
{\small Via Dodecaneso 33, 16146 Genova, Italy.}\vspace{2mm}\\
{\small ${}^3$ School of Mathematical Sciences, University of Nottingham,}\\
{\small University Park, Nottingham NG7 2RD, United Kingdom.}\vspace{4mm}\\
{\small \begin{tabular}{ll}
Email: & ${}^a$~\href{mailto:marco.benini@unige.it}{\texttt{marco.benini@unige.it}}\\
&${}^b$~\href{mailto:alastair.grant-stuart@nottingham.ac.uk}{\texttt{alastair.grant-stuart@nottingham.ac.uk}}\\
& ${}^c$~\href{mailto:giorgio.musante@edu.unige.it}{\texttt{giorgio.musante@edu.unige.it}}\\
& ${}^d$~\href{mailto:alexander.schenkel@nottingham.ac.uk}{\texttt{alexander.schenkel@nottingham.ac.uk}}
\vspace{2mm}
\end{tabular}
}
}

\date{November 2025}

\begin{document}

\maketitle

\begin{abstract}
\noindent M{\o}ller maps are identifications between the observables of a perturbatively interacting physical system and the observables of its underlying free (i.e.\ non-interacting) system. This work studies and characterizes obstructions to the existence of such identifications. The main results are existence and importantly also non-existence theorems, which in particular imply that M{\o}ller maps \textit{do not} in general exist for non-Abelian Chern-Simons and Yang-Mills theories on globally hyperbolic Lorentzian manifolds. These results are obtained through homological algebra techniques which are of independent interest in the analysis of classical field theories.
\end{abstract}
\vspace{-1mm}

\paragraph*{Keywords:} perturbative field theory, M{\o}ller maps, BV formalism, $L_\infty$-algebras
\vspace{-2mm}

\paragraph*{MSC 2020:} 70Sxx, 81Txx, 55Uxx
\vspace{-2mm}

\tableofcontents

\newpage 


\section{\label{sec:intro}Introduction and summary}
M{\o}ller maps are identifications between the observables of a 
perturbatively interacting physical system and the observables
of its underlying free (i.e.\ non-interacting) system. In the context
of field theories on globally hyperbolic Lorentzian manifolds, they attempt
to formalize the idea that it should be possible to identify, say 
in the far past, the free with the perturbatively interacting observables of a theory, and 
thereby obtain a global identification of observables on spacetime 
via the time-slice axiom. 
\sk

In this paper we will perform a systematic study of 
M{\o}ller maps for classical physical systems, potentially with gauge symmetries,
which we formulate within the BV formalism or dually in terms of $L_\infty$-algebras. The classical
observables for such a system thus form a commutative differential graded algebra 
(over the ring of power series in a formal parameter $\lambda$, interpreted as coupling constant) whose 
differential $\delta_0 + \delta$ is a linear combination
of the differential $\delta_0$ of the underlying free system and a $\lambda$-formal perturbation $\delta$ 
which encodes the non-linear features, such as non-Abelian gauge symmetries and perturbative interactions.
In this setting we introduce an axiomatic concept of M{\o}ller maps (Definition \ref{def:Moeller})
as cochain isomorphisms which intertwine between the free differential $\delta_0$ and the 
perturbatively interacting one $\delta_0 + \delta$. 
Loosely speaking, this means that such M{\o}ller maps intertwine between 
the free equation of motion and gauge symmetry 
and the perturbatively interacting ones. 
We will then show in Proposition \ref{propo:obstructions}
that there is a tower of successive cohomological obstructions to the existence of M{\o}ller maps.
In particular, their existence is \textit{not} automatically guaranteed by the Maurer-Cartan equation for the formal 
perturbation $\delta$ or, in BV terminology, the classical master equation of the theory.
We would like to stress that this should not come as a surprise:
Perturbing the differential of a cochain complex alters in general its cohomology, 
even when the perturbation is formal. In the latter case, 
homological algebra provides efficient tools, known as homological perturbation theory \cite{HPT}, 
which allow one to quantify the extent to which 
the cohomology is altered. Therefore, from a homological algebra 
perspective, one expects the existence of M{\o}ller maps only 
under quite restrictive assumptions.
\sk

The main results of this paper are theorems about the existence and non-existence of M{\o}ller maps for 
specific classes of theories. In order to simplify our presentation, 
we will first state and prove these theorems for finite-dimensional
systems in Section \ref{sec:obstructions} and then generalize to the case of classical field theories
in Section \ref{sec:fieldtheories}. Loosely speaking, our existence theorem (Theorem \ref{theo:existence}) 
proves that a M{\o}ller map exists for theories that have cohomological features which are characteristic
for interacting Klein-Gordon fields on globally hyperbolic Lorentzian manifolds. 
This theorem generalizes, under slightly refined hypotheses, to the case of field theories 
(Theorem \ref{theo:FTexistence}) and it provides explicit M{\o}ller maps
for interacting Klein-Gordon theories (Example \ref{ex:KGexistence}).
These agree with earlier constructions in the context of perturbative algebraic (quantum) field theory,
see in particular \cite{Duetsch,Hawkins,RejznerProceedings}.
\sk

Our non-existence theorem (Theorem \ref{theo:nonexistence}) is more interesting.
Loosely speaking, it proves that M{\o}ller maps \textit{do not} exist for theories that have
cohomological features which are characteristic for non-Abelian gauge theories
treated perturbatively around the trivial background solution. (It is worthwhile to
note that these features may change for perturbations around
non-trivial background solutions, see Remark \ref{rem:background}.)
This theorem generalizes, under slightly refined hypotheses, to the case of field theories 
(Theorem \ref{theo:FTnonexistence}) and it provides non-existence results for M{\o}ller
maps for non-Abelian Chern-Simons theories (Example \ref{ex:CSnonexistence})
and non-Abelian Yang-Mills theories (Example \ref{ex:YMnonexistence}), treated perturbatively
around the trivial background solution. As we explain in more detail in 
Remark \ref{rem:origin}, these obstructions to the existence of M{\o}ller maps are of a gauge theoretic origin, 
and not of a dynamical one. They arise from the interplay between the stabilizers 
of the linearized gauge transformations and the non-trivial Lie algebra structure on 
the infinitesimal non-Abelian gauge transformations.
These obstructions are therefore detected in the degree $-1$ cohomology of the field complex, i.e.\ on ghosts.
We would like to emphasize that these are certainly not the only obstructions to the existence of M{\o}ller maps.
To illustrate this latter point,
we exhibit in Remark \ref{rem:CSlastexample} a simple gauge theory model
(non-Abelian Chern-Simons theory on $\mathbb{R}^2 \times \mathbb{S}^1$)
that has perturbatively interacting observables which differ explicitly from its free ones in degree $0$.
\sk

The outline of the remainder of this paper is as follows: In Section \ref{sec:convections}
we recall some basic concepts from homological algebra in order to fix our notation and conventions.
Section \ref{sec:BVformalism} provides a concise review of the BV formalism and its dual language of 
$L_\infty$-algebras. These concepts will be illustrated by three main examples, given by
perturbatively interacting Klein-Gordon, non-Abelian Chern-Simons and non-Abelian Yang-Mills theories.
In Section \ref{sec:obstructions} we study the concept of M{\o}ller maps (Definition \ref{def:Moeller})
in a finite-dimensional setting, characterize cohomological obstructions to their existence 
(Proposition \ref{propo:obstructions}), and prove our existence and non-existence theorems 
(Theorems \ref{theo:existence} and \ref{theo:nonexistence}). In Section \ref{sec:fieldtheories}
we generalize these results to the case of classical field theories 
(Theorems \ref{theo:FTexistence} and \ref{theo:FTnonexistence}), which requires
slightly refined and more constructive hypotheses in order to deal with infinite-dimensional phenomena. 
As an application, we construct M{\o}ller maps for perturbatively interacting
Klein-Gordon theories (Example \ref{ex:KGexistence}) and prove non-existence
of M{\o}ller maps for non-Abelian Chern-Simons and Yang-Mills theories 
treated perturbatively around the trivial background solution
(Examples \ref{ex:CSnonexistence} and \ref{ex:YMnonexistence}).

\paragraph{Comparison to perturbative algebraic (quantum) field theory:}
Perturbative algebraic quantum field theory (pAQFT) 
encompasses a collection of approaches whose aim is
to construct perturbatively interacting quantum field theories
on globally hyperbolic Lorentzian manifolds by using concepts
and techniques which are inspired by AQFT. Gauge theoretic models
in pAQFT are typically described by using the BV formalism,
see e.g.\ \cite{Hollands,Zahn} for Yang-Mills theories and \cite{FR1,FR2,RejznerBook}
for more general classes of gauge theories. 
\sk

A common feature of pAQFT is that the perturbatively interacting observables
of a theory are expressed in terms of formal power series 
(in $\hbar$ and the coupling constants $\lambda$) of 
observables of the underlying free theory. However,
there exist different proposals in the literature for
how to define the interacting BV differential.
One of these proposals is to use quantum M{\o}ller maps
in order to transfer the free BV differential to an interacting one,
see e.g.\ \cite[Eqn.~(40)]{FR2} and \cite[Chapter 7.3]{RejznerBook},
while others provide constructions of the interacting BV differential
which do not seem to be directly related to M{\o}ller maps, see e.g.\
\cite[Theorem 10]{Frob}. See also \cite[Section V.B]{TaslimiTehrani} 
and the remark after \cite[Theorem 10]{Frob} for further comments
about differences in the existing approaches to pAQFT.
\sk

The content and results of our paper are most directly linked to those parts of
the pAQFT literature in which these techniques are applied to the
simpler context of \textit{classical} field theories, where one can 
avoid the sophisticated and involved renormalization techniques from
the quantum case. Classical M{\o}ller maps have appeared before
for instance in \cite[Section 2.1]{Duetsch}, \cite[Section 3.2]{Hawkins}
and \cite[Section 2.5]{RejznerProceedings}. In the context of
the (classical) BV formalism, they are used to intertwine between
the free and interacting BV differentials, see e.g.\ \cite[Corollary 2.16]{RejznerProceedings},
which is precisely the defining property of M{\o}ller maps 
that we axiomatize in Definitions \ref{def:Moeller} and \ref{def:Moeller-cont}.
In particular, the classical M{\o}ller maps constructed 
in these works are particular examples of our axiomatic concept of M{\o}ller maps.
\sk

The main difference between our approach and the pA(Q)FT literature
is that we work exclusively with the BV differential $\{S_{\mathrm{BV}},-\}$ 
while in pA(Q)FT one manifestly uses also the BV action
$S_{\mathrm{BV}}=\int_M \mathcal{L}_{\mathrm{BV}}$. 
As a consequence of locality of the antibracket $\{-,-\}$,
the BV differential is associated directly to the Lagrangian density 
$\mathcal{L}_{\mathrm{BV}}$, hence it is well-defined on
a (necessarily non-compact) globally hyperbolic Lorentzian manifold $M$.
In stark contrast to this, the BV action requires an integration over the spacetime
$M$, hence it is ill-defined on globally hyperbolic Lorentzian manifolds.
The practical solution to this issue which is provided in the pA(Q)FT literature
is to consider suitably regularized BV actions, such as e.g.\
$S_{\mathrm{BV}}(f) = \int_M f\,\mathcal{L}_{\mathrm{BV}}$ where $f\in C^\infty_\cc(M)$
is a compactly supported cutoff function. An explicit 
candidate for a (classical or quantum) M{\o}ller map is then
constructed from these regularized BV actions. In the classical case,
one can directly verify (see e.g.\
\cite[Theorem 2.15]{RejznerProceedings}) that this map has the desired intertwining property 
between the free and interacting BV differentials that we axiomatize in Definitions 
\ref{def:Moeller} and \ref{def:Moeller-cont}.
It is important to emphasize that such regularized M{\o}ller maps
depend on the choice of cutoff $f\in C^\infty_\cc(M)$ and they
provide an identification between the free observables and the 
perturbatively interacting ones which correspond to the regularized
BV action $S_{\mathrm{BV}}(f)$,
while our non-existence theorem (Theorem \ref{theo:FTnonexistence}) and the 
associated Examples \ref{ex:CSnonexistence} and \ref{ex:YMnonexistence} deal with
perturbative gauge field theories without cutoffs.
\sk

One possible physical interpretation of our results is as follows:
The regularized M{\o}ller maps from pA(Q)FT will in general be 
obstructed in the `adiabatic limit' $f\to 1$ in 
which the cutoff in the BV action is removed. More specifically, 
this limit is unobstructed for examples such as interacting Klein-Gordon theories 
for which our existence theorem (Theorem \ref{theo:FTexistence}) and Example 
\ref{ex:KGexistence} apply. In contrast to this,
our non-existence theorem (Theorem \ref{theo:FTnonexistence}) in combination 
with Examples \ref{ex:CSnonexistence} and \ref{ex:YMnonexistence}
shows that the regularized M{\o}ller maps 
for non-Abelian Chern-Simons and Yang-Mills theories, 
treated perturbatively around the trivial background solution, must be obstructed in this limit. 
Since quantum M{\o}ller maps should specialize to classical ones when sending $\hbar\to 0$,
we expect that these obstructions also exist for the quantum M{\o}ller maps from pAQFT.
While our results do not imply any internal inconsistencies of the pAQFT framework, 
they suggest that taking the `adiabatic limit' $f\to 1$ could be much 
more subtle for gauge theories than it is for field theories without gauge symmetries.
We think that it would be interesting and fruitful to carry out a detailed study of such
`adiabatic limits', which due to the observations in Remark \ref{rem:background} may 
depend strongly on the background solution around which one perturbs. Unfortunately, 
we are currently unable to find a suitable axiomatization of 
`adiabatic limits' in our homological algebraic setting and hence we 
have to leave this problem open for future works.

\section{\label{sec:convections}Notation and conventions} 
All vector spaces in this work are over a field $\bbK$ which is either
the real numbers $\bbK=\bbR$ or the complex numbers $\bbK=\bbC$.
We denote by $\Ch$ the category of 
cochain complexes of $\bbK$-vector spaces. An object in this category is a cochain complex $V = (V,\dd)$,
i.e.\ a family $V=\{V^i\}_{i\in\bbZ}$ of $\bbK$-vector spaces indexed by their cohomological 
degree $i\in\bbZ$, together with a family $\dd = \{\dd:V^i\to V^{i+1}\}_{i\in\bbZ}$ of degree-increasing linear maps 
(called the differential) which squares to zero $\dd^2=0$. A morphism in $\Ch$ is a cochain map
$f : V\to W$, i.e.\ a family $f = \{f: V^i\to W^i\}_{i\in\bbZ}$ of degree-preserving linear maps
which commutes with the differentials $\dd\,f=f\,\dd $.
\sk

Let us recall that the category $\Ch$ is bicomplete, i.e.\ all small limits and
colimits exist, and closed symmetric monoidal. The tensor product $V\otimes W\in\Ch$ 
of two cochain complexes $V,W\in\Ch$ is defined by setting
\begin{subequations}\label{eqn:otimes}
\begin{flalign}
(V\otimes W)^i\,:=\, \bigoplus_{j\in\bbZ }\left(V^j\otimes W^{i-j}\right)\quad,
\end{flalign}
for all $i\in\bbZ$, and the differential is given by the Leibniz rule
\begin{flalign}
\dd(v\otimes w) \,:=\, (\dd v)\otimes w + (-1)^{\vert v\vert}\,v\otimes (\dd w)\quad,
\end{flalign}
\end{subequations}
for all homogeneous $v\in V$ and $w\in W$, where $\vert v\vert\in\bbZ$ denotes the cohomological degree.
The monoidal unit for this tensor product is given by regarding the field $\bbK = (\bbK,0)\in\Ch$
as a cochain complex concentrated in degree $0$ with trivial differential. The symmetric braiding
is defined by the Koszul sign rule
\begin{flalign}
\gamma\,:\, V\otimes W~\longrightarrow~W\otimes V~~,\quad
v\otimes w~\longmapsto~(-1)^{\vert v\vert\,\vert w\vert}\, w\otimes v\quad,
\end{flalign}
for all homogeneous  $v\in V$ and $w\in W$.
\sk

The internal hom complex $\hom(V,W) \in\Ch$ between two cochain complexes $V,W\in \Ch$ is defined by
\begin{subequations}
\begin{flalign}
\hom(V,W)^i\,:=\,\prod_{j\in\bbZ}\Hom_{\bbK}\big(V^j,W^{j+i}\big)\quad,
\end{flalign}
for all $i\in\bbZ$, where $\Hom_\bbK$ denotes the vector space of linear maps, 
and the `adjoint' differential
\begin{flalign}
\partial(L)\,:=\,\dd\, L - (-1)^{\vert L\vert}\, L\,\dd\quad,
\end{flalign}
\end{subequations}
for all homogeneous $L\in\hom(V,W)$. Note that cochain maps $f:V\to W$ are precisely the $0$-cocycles
$f\in \mathsf{Z}^0\,\hom(V,W)$ in the internal hom complex, i.e.\ $f\in \hom(V,W)^0$ such that 
$\partial(f)=0$. 
A cochain homotopy between two cochain maps $f:V\to W$ and $g:V\to W$ is a $(-1)$-cochain 
$h\in\hom(V,W)^{-1}$ such that $\partial(h) = g-f$. Cochains of degree $-2$ and lower in $\hom(V,W)$
describe higher cochain homotopies.
\sk

Our convention for shifts of cochain complexes is as follows. Given any cochain complex $V = (V,\dd)\in\Ch$
and any integer $p\in\bbZ$, we define the $p$-shifted cochain complex $V[p] = (V[p],\dd_{[p]})\in\Ch$
by $V[p]^i := V^{i+p}$, for all $i\in\bbZ$, and $\dd_{[p]} := (-1)^{p}\,\dd$. Note that
$V[0] = V$ and $V[p][q] = V[p+q]$, for all $p,q\in\bbZ$. One can describe $V[p]\in \Ch$ equivalently
in terms of the tensor product $\bbK[p]\otimes V\in\Ch$ of cochain complexes, where $\bbK[p]\in\Ch$
is the cochain complex given by the field $\bbK$ concentrated in degree $-p$ with trivial differential.
\sk

A commutative differential graded algebra (in short, a commutative dg-algebra) is
a commutative, associative and unital algebra $A = (A,\mu,\eta)$ in the symmetric monoidal category $\Ch$.
The $\Ch$-morphisms $\mu:A\otimes A\to A$ and $\eta : \bbK\to A$ are called, respectively, 
the multiplication and unit. We denote by $\dgCAlg$ the category whose objects are
commutative dg-algebras and morphisms $f: A\to B$ are cochain maps which respect the multiplications and units,
i.e.\ $\mu_B\,(f\otimes f) = f\,\mu_A$ and $\eta_B = f\,\eta_A$. A basic, but important, example is given by
the free commutative dg-algebra associated with a cochain complex $V\in\Ch$, i.e.\ the symmetric algebra
\begin{subequations}\label{eqn:Sym}
\begin{flalign}
\Sym\,V\,=\,\big(\Sym \,V,\,\mu,\,\eta\big)\,\in\dgCAlg\quad.
\end{flalign}
In more detail, the underlying cochain complex is defined by
\begin{flalign}
\Sym\,V\,:=\,\bigoplus_{n=0}^{\infty} \Sym^n\,V \,:=\, \bigoplus_{n=0}^\infty \big(V^{\otimes n}\big)_{\Sigma_n}\,\in\,\Ch\quad,
\end{flalign}
\end{subequations}
where $\Sym^n\,V:=\big(V^{\otimes n}\big)_{\Sigma_n}\in\Ch$ denotes the coinvariants (i.e.\ quotient) 
of the action of the permutation group $\Sigma_n$ on tensor powers $V^{\otimes n}$ via
the symmetric braiding $\gamma$. The multiplication reads explicitly as $\mu : \Sym\,V\otimes \Sym\,V\to\Sym\,V\,,~
[v_1\otimes \cdots \otimes v_n]\otimes [v_{n+1}\otimes \cdots\otimes v_{n+m}]
\mapsto [v_1\otimes \cdots\otimes v_{n+m}]$ and the unit $\eta : \bbK\to \Sym\,V$
is given by the inclusion $\bbK\cong \Sym^0\,V\hookrightarrow \Sym\,V$. Note that the
differential $\dd$ on $\Sym\,V$ agrees, by construction, on the generators $\Sym^1\,V = V$
with the differential $\dd$ on $V\in\Ch$. 
As a consequence of the Leibniz rule $\dd(a\,b) = (\dd a)\,b +(-1)^{\vert a\vert}\,a\,(\dd b)$
for the multiplication $a\,b := \mu(a\otimes b)$, this completely fixes $\dd$ on all of $\Sym\,V$.


\section{\label{sec:BVformalism}\texorpdfstring{BV formalism and $L_\infty$-algebras}{BV formalism and L-infinity-algebras}}
In this section we briefly recall the homological description of classical physical systems, 
such as classical field theories, in the Batalin-Vilkovisky (BV) formalism and its
dual language of $L_\infty$-algebras. For details we refer the reader to the books
\cite{CG1,CG2} of Costello and Gwilliam, and the informative article \cite{Jurco} of 
Jur\v{c}o, Raspollini, S\"amann and Wolf.
\sk

In this formalism, a free (i.e.\ non-interacting) physical system 
is described by a cochain complex
\begin{flalign}\label{eqn:FFFcomplex}
\FFF\,=\,\Big(\xymatrix{
\cdots \ar[r]^-{\dd} \,&\, \FFF^{-1} \ar[r]^-{\dd} \,&\, \FFF^0  \ar[r]^-{\dd}\,&\,
\FFF^{1}\ar[r]^-{\dd} \,&\, \FFF^{2}\ar[r]^-{\dd} \,&\,\cdots
}
\Big)\,\in\,\Ch\quad.
\end{flalign}
The components of this complex describe the fields $\FFF^0$, the ghosts $\FFF^{<0}$ and the antifields $\FFF^{>0}$
of the theory, and the differential $\dd$ encodes the linear dynamics and gauge symmetries. For physical
systems which arise from extremizing an action functional, the complex $\FFF$ comes
further endowed with a canonical $(-1)$-shifted symplectic structure, whose dual shifted 
Poisson structure is the antibracket from the BV formalism. We would like to stress that
this is a semi-classical datum whose relevance is for (BV) quantization, hence it 
will play no role in our discussion of classical physical systems in this paper.
\sk

Let us provide some illustrative examples.
\begin{ex}\label{ex:KG}
The free Klein-Gordon field on a globally hyperbolic Lorentzian manifold $M$ is described by the 
$2$-term cochain complex
\begin{flalign}\label{eqn:KG}
\FFF_{\mathrm{KG}}\,:=\,\Big(\xymatrix@C=3em{
\stackrel{(0)}{C^\infty(M)} \ar[r]^-{\square + m^2} \,&\,\stackrel{(1)}{C^\infty(M)}
}\Big)
\end{flalign}
which is concentrated in degrees $\{0,1\}$. Elements $\Phi\in C^\infty(M) $ in degree $0$
are interpreted as the fields and elements $\Phi^\ddagger \in C^\infty(M)$ in degree $1$ are the antifields.
The only non-vanishing component of the differential is the Klein-Gordon operator $\square + m^2$.
The $0$-th cohomology $\mathsf{H}^0\FFF_{\mathrm{KG}}$ of this complex is the space 
of solutions of the Klein-Gordon equation $(\square + m^2)\,\Phi=0$.
The degree $1$ cohomology $\mathsf{H}^1\FFF_{\mathrm{KG}}\cong 0$ 
is trivial because, due to global hyperbolicity of $M$, 
the inhomogeneous Klein-Gordon equation $(\square + m^2)\,\Phi=\Phi^\ddagger$ admits a solution 
$\Phi\in C^\infty(M)$ for any right-hand side $\Phi^\ddagger\in C^\infty(M)$.
\end{ex}
\begin{ex}\label{ex:CS}
Let $\g$ be a semi-simple Lie algebra. The underlying free theory
of the associated non-Abelian Chern-Simons field on a $3$-manifold $M$ is described by the cochain complex
\begin{flalign}\label{eqn:CS}
\FFF_{\mathrm{CS}}\,:=\,\Big(\xymatrix{
\stackrel{(-1)}{\Omega^0(M,\g)} \ar[r]^-{\dd_{\dR}}\,&\,\stackrel{(0)}{\Omega^1(M,\g)}
\ar[r]^-{\dd_{\dR}} \,&\, \stackrel{(1)}{\Omega^2(M,\g)}\, \ar[r]^-{\dd_{\dR}}
&\, \stackrel{(2)}{\Omega^3(M,\g)}
}\Big)
\end{flalign}
which is concentrated in degrees $\{-1,0,1,2\}$,
where $\Omega^\bullet(M,\g):= \Omega^\bullet(M)\otimes \g$ denotes the $\g$-valued differential forms on $M$. 
Elements $c\in \Omega^0(M,\g)$ in degree $-1$ are interpreted as the ghosts,
elements $A\in \Omega^1(M,\g)$ in degree $0$ are the gauge fields,
elements $A^\ddagger\in \Omega^2(M,\g)$ in degree $1$ are the antifields
and elements $c^\ddagger\in\Omega^3(M,\g)$ in degree $2$ are the antifields 
for ghosts. The differential is given by the de Rham differential,
which encodes in particular the action of 
linearized gauge transformations $A\mapsto A+\dd_\dR c$
and the linearized equation of motion (i.e.\ flatness condition) $\dd_\dR A =0$. The cohomology
of this complex is given by the tensor product
$\mathsf{H}^\bullet\FFF_{\mathrm{CS}}
= \mathsf{H}_{\dR}^{\bullet+1}(M)\otimes\g$ of the shifted de Rham cohomology of $M$ with $\g$,
hence it depends on the shape of the manifold $M$. Note
that the degree $-1$ cohomology $\mathsf{H}^{-1}\FFF_{\mathrm{CS}}\cong \g^{\pi_0(M)}$
is given by locally-constant $\g$-valued $0$-forms, i.e.\ $c\in \Omega^0(M,\g)$ such that
$\dd_\dR c=0$, hence it can {\em not} be trivial since every non-empty manifold has a non-empty set
of connected components $\pi_0(M)\neq \varnothing$.
This cohomology describes the stabilizers $A\mapsto A + \dd_\dR c = A$
of the linearized gauge transformations.
\end{ex}
\begin{ex}\label{ex:YM}
Let $\g$ be a semi-simple Lie algebra. The underlying free theory
of the associated non-Abelian Yang-Mills field on an $m$-dimensional 
globally hyperbolic Lorentzian manifold $M$ 
is described by the cochain complex
\begin{flalign}\label{eqn:YM}
\FFF_{\mathrm{YM}}\,:=\,\Big(\xymatrix@C=2em{
\stackrel{(-1)}{\Omega^0(M,\g)} \ar[r]^-{\dd_{\dR}}\,&\,\stackrel{(0)}{\Omega^1(M,\g)}
\ar[rr]^-{\dd_{\dR}\,\ast\,\dd_{\dR}} \,&\,\,&\, \stackrel{(1)}{\Omega^{m-1}(M,\g)}\, \ar[r]^-{\dd_{\dR}}
&\, \stackrel{(2)}{\Omega^m(M,\g)}
}\Big)
\end{flalign}
which is concentrated in degrees $\{-1,0,1,2\}$, where $\ast: \Omega^\bullet(M,\g)\to \Omega^{m-\bullet}(M,\g)$
denotes the Hodge operator of $M$. Elements $c\in \Omega^0(M,\g)$ in degree $-1$ are interpreted as the ghosts,
elements $A\in \Omega^1(M,\g)$ in degree $0$ are the gauge fields,
elements $A^\ddagger\in \Omega^{m-1}(M,\g)$ in degree $1$ are the antifields
and elements $c^\ddagger\in\Omega^m(M,\g)$ in degree $2$ are the antifields 
for ghosts. The differential is constructed from the de Rham differential and the Hodge operator.
It encodes in particular the action of 
linearized gauge transformations $A\mapsto A+\dd_\dR c$
and the linearized Yang-Mills equation $\dd_\dR \ast \dd_\dR A =0$.
The $0$-th cohomology $\mathsf{H}^0\FFF_{\mathrm{YM}}$ 
of this complex is the space of solutions of the linearized Yang-Mills equation
modulo linearized gauge transformations. As for the free Chern-Simons theory
from Example \ref{ex:CS}, the degree $-1$ cohomology 
$\mathsf{H}^{-1}\FFF_{\mathrm{YM}}\cong \g^{\pi_0(M)}$
is given by locally-constant $\g$-valued $0$-forms, hence it can
{\em not} be trivial. The origin of this non-trivial cohomology lies again in
the existence of non-trivial stabilizers $A\mapsto A + \dd_\dR c = A$  of the linearized gauge transformations.
\end{ex}

\begin{rem}\label{rem:background}
The examples presented above can be interpreted as
arising from linearizations of non-linear field theories
around the trivial solution. To generalize these examples, 
one may also consider linearizations
around a non-trivial solution of the non-linear field equations. 
We will explain below that such
modifications may alter the cohomology of the free field complex.
\sk

For an interacting Klein-Gordon field, the generalization 
to non-trivial backgrounds amounts to replacing 
the mass parameter $m^2$ in \eqref{eqn:KG} with a spacetime-dependent 
expression $m^2(x)$ which is determined by the background solution. 
The resulting differential operator $\square + m^2(x)$ is normally hyperbolic,
hence the cohomology is concentrated in degree $0$ by the same arguments as in Example \ref{ex:KG}.
Using well-posedness of the Cauchy problem, one finds that the cohomology does not
depend on $m^2(x)$ as it is isomorphic to the space of initial data on a Cauchy surface.
\sk

For non-Abelian Chern-Simons theory, the generalization 
to non-trivial backgrounds amounts to replacing
the de Rham differential in \eqref{eqn:CS} 
with the twisted differential $\dd_A := \dd_\dR + [A,\,\cdot\,]$
where $A\in\Omega^1(M,\g)$ is any flat background connection, i.e.\ the curvature
$F_A = \dd_{\dR}A +\frac{1}{2}[A,A]=0$ vanishes. The resulting
cohomology is then given by the cohomology of the twisted differential.
In particular, the degree $-1$ cohomology 
consists of sections $c \in \Omega^0(M,\g)$  
which are flat $\dd_A c =0$ with respect to the 
background flat connection. The existence of
such degree $-1$ cohomology classes is in general
obstructed by global features of the flat background 
connection $A$, which are encoded in its holonomies.
\sk

For non-Abelian Yang-Mills theory, the generalization 
to non-trivial backgrounds amounts to twisting
the de Rham differential in \eqref{eqn:YM} with a background 
connection $A\in\Omega^1(M,\g)$ which satisfies the non-Abelian 
Yang-Mills equation $\dd_A \ast F_A=0$. In this case the
existence of degree $-1$ cohomology classes, 
i.e.\ sections $c \in \Omega^0(M,\g)$ which are flat $\dd_A c =0$
with respect to the background Yang-Mills connection,
is already locally obstructed by the curvature of $A$.
In a similar spirit, one may consider a Yang-Mills-Higgs system 
and linearize around a background solution with a non-trivial Higgs field,
which yields further obstructions to the 
existence of degree $-1$ cohomology classes.
\end{rem}

To incorporate perturbative interactions into this picture, one endows the $(-1)$-shifted cochain complex
\begin{flalign}
\FFF[-1]\,=\, \Big(\xymatrix{
\cdots \ar[r]^-{-\dd} \,&\, \stackrel{(0)}{\FFF^{-1}} \ar[r]^-{-\dd} \,&\, \stackrel{(1)}{\FFF^0} 
\ar[r]^-{-\dd}\,&\, \stackrel{(2)}{\FFF^{1}}\ar[r]^-{-\dd} \,&\, \stackrel{(3)}{\FFF^{2}}\ar[r]^-{-\dd} \,&\,\cdots
}
\Big)
\end{flalign}
with an $L_\infty$-algebra structure which extends its differential. 
Let us recall that the latter is a family of linear maps
\begin{subequations}\label{eqn:Linftyshifted}
\begin{flalign}
\ell^{[-1]}\,=\,\Big\{\ell^{[-1]}_n\,\in\, 
\hom\big(\textstyle{\bigwedge^n} (\FFF[-1]),\,\FFF[-1]\big)^{2-n}\Big\}_{n\in\bbN^{\geq 1}}
\end{flalign}
from the exterior powers of the $(-1)$-shifted complex $\FFF[-1]$, 
with 
\begin{flalign}
\ell^{[-1]}_1 \,:\,= -\dd 
\end{flalign}
\end{subequations} 
the given differential of $\FFF[-1]$ and $\ell^{[-1]}_n$ a linear map of degree $2-n$,
which satisfies the homotopy Jacobi identities, see e.g.\ \cite{Jurco} for explicit formulas.
We prefer to work with the alternative, but equivalent, `symmetric' convention
in which an $L_\infty$-algebra structure is given by a family of degree $1$ linear maps
\begin{subequations}\label{eqn:Linfty}
\begin{flalign}
\ell\,=\,\Big\{\ell_n\,\in\, \hom\big(\Sym^n \,\FFF,\,\FFF\big)^{1}\Big\}_{n\in\bbN^{\geq 1}}
\end{flalign}
from the symmetric powers of the unshifted complex $\FFF$, with 
\begin{flalign}
\ell_1 \,:=\, \dd 
\end{flalign}
\end{subequations} 
the given differential of $\FFF$. To avoid confusion, we have to 
distinguish between $L_\infty$-algebra structures in the antisymmetric 
\eqref{eqn:Linftyshifted} and symmetric \eqref{eqn:Linfty} conventions.
For this we follow standard practice and call the latter $L_\infty[1]$-algebras, 
see e.g.\ \cite{KraftSchnitzer}.
The equivalence between $L_\infty$-algebra structures \eqref{eqn:Linftyshifted} 
and $L_\infty[1]$-algebra structures \eqref{eqn:Linfty} is given by the cochain isomorphism
\begin{flalign}
\nn\hom\big(\textstyle{\bigwedge^n} (\FFF[-1]),\,\FFF[-1]\big)\,&\cong\,
\hom\big(\big(\Sym^n\, \FFF\big)[-n],\,\FFF[-1]\big)\\
\,&\cong\,
\hom\big(\Sym^n\, \FFF,\,\FFF[n-1]\big)\,\cong\,
\hom\big(\Sym^n\, \FFF,\,\FFF\big)[n-1]\quad,
\end{flalign}
see e.g.\ \cite[Appendix A]{Jurco} or \cite[Section 3]{KraftSchnitzer}
for the details. The homotopy Jacobi identities for $L_\infty[1]$-algebras are given by
\begin{flalign}\label{eqn:homotopyJacobi}
\sum_{i+j=n+1} \ell_{i}\,\big(\ell_j\otimes \id^{\otimes {i-1}}\big)\!\!\!\sum_{\sigma\in\mathrm{Sh}(j,i-1)}\!\!\! \gamma_{\sigma}~=~0\quad,\qquad\text{for all }n\in\bbN^{\geq 1}\quad,
\end{flalign}
where $\mathrm{Sh}(j,i-1)\subseteq\Sigma_{n}$ denotes the $(j,i-1)$-shuffle permutations
and $\gamma_{\sigma}: \FFF^{\otimes n}\to \FFF^{\otimes n}$ denotes the action 
of the permutation $\sigma$ on tensor powers $\FFF^{\otimes n}$
via the symmetric braiding $\gamma$ of $\Ch$.
\begin{ex}\label{ex:KG2}
Recall the free Klein-Gordon complex $\FFF_{\mathrm{KG}}$ from Example \ref{ex:KG}.
The $L_\infty[1]$-algebra structure corresponding to $\Phi^N$-theory is given by
\begin{subequations}\label{eqn:KGbrackets}
\begin{flalign}
\ell_n \,= \, 0\quad,\qquad \text{for all }2\leq n\neq N-1\quad,
\end{flalign}
and
\begin{flalign}
\ell_{N-1}\big(\Phi_1,\cdots,\Phi_{N-1}\big) \,:=\, \Phi_1\cdots\Phi_{N-1} \,\in\,
\FFF^1_{\mathrm{KG}}\,=\,C^\infty(M)\quad,
\end{flalign}
\end{subequations}
for all degree $0$ elements $\Phi_1,\dots,\Phi_{N-1}\in \FFF_\mathrm{KG}^0=C^\infty(M)$.
(Note that the degree $1$ linear map 
$\ell_{N-1}\in\hom\big(\Sym^{N-1}\, \FFF_{\mathrm{KG}},\,\FFF_{\mathrm{KG}}\big)^1 $
must vanish for degree reasons whenever at least one of 
the $N-1$ inputs has positive cohomological degree $\FFF^{1}_{\mathrm{KG}}$.)
The homotopy Jacobi identities \eqref{eqn:homotopyJacobi} are trivially satisfied in this example.
\sk

The reader might wonder why a non-trivial $\ell_{N-1}$ bracket should correspond to $\Phi^N$-theory.
This shift by $-1$ in the field polynomial degree 
comes from the fact that the $L_\infty[1]$-algebra structure encodes the deformation
of the equation of motion, and not directly of an action functional. The equation of motion
for $\Phi^N$-theory has a non-linear term of the form $\Phi^{N-1}$, which is obtained by
evaluating the bracket $\ell_{N-1}$ on $N-1$ copies of the same field $\Phi\in C^\infty(M)$, 
i.e.\ $\ell_{N-1}(\Phi,\dots,\Phi) = \Phi^{N-1}$. 
\sk

The link between $L_\infty$-algebra structures, equations of motion 
and action functionals can be understood best by using
the canonical $(-1)$-shifted symplectic structure 
$\langle -,-\rangle : \FFF_{\mathrm{KG},\cc}\otimes \FFF_{\mathrm{KG},\cc}\to \bbR[-1]\,,
~\Phi\otimes\Phi^\ddagger\mapsto\int_M\Phi\,\Phi^\ddagger\,\vol_M$, where
$\FFF_{\mathrm{KG},\cc} \subseteq \FFF_{\mathrm{KG}}$ denotes the subcomplex
of compactly supported sections,  
in order to introduce the Maurer-Cartan action
\begin{flalign}
\nn S[\Phi] \,&=\, \sum_{n=1}^\infty \tfrac{1}{(n+1)!}\,\langle \Phi, \ell_n(\Phi,\dots,\Phi)\rangle\\
\,&=\,\int_M \Big(\tfrac{1}{2!}\,\Phi (\square+m^2)\Phi + \tfrac{1}{3!}\, \Phi\,\ell_2(\Phi,\Phi) + \cdots\Big)\,\vol_M\quad.
\end{flalign}
Note that for the choice of brackets in \eqref{eqn:KGbrackets} this yields the action of $\Phi^N$-theory.
The Euler-Lagrange equation for this action is given by the Maurer-Cartan equation
$\sum_{n=1}^\infty \tfrac{1}{n!}\,\ell_n(\Phi,\dots,\Phi)=0$,
which in our example specializes to the equation of motion for $\Phi^N$-theory
as observed in the previous paragraph.
We refer the reader to \cite[Section 4]{Jurco} for more details and further examples.
\end{ex}
\begin{ex}\label{ex:CS2}
Recall the free Chern-Simons complex $\FFF_{\mathrm{CS}}$ from Example \ref{ex:CS}.
The $L_\infty[1]$-algebra structure corresponding to non-Abelian Chern-Simons theory is given
by
\begin{subequations}
\begin{flalign}
\ell_n \,=\,0\quad,\qquad\text{for all }n\geq 3\quad,
\end{flalign}
and
\begin{flalign}
\ell_2\big(\alpha_1\otimes t_1,\alpha_2\otimes t_2\big)\,:=\,
(-1)^{\vert\alpha_1\vert}~\big[\alpha_1\otimes t_1 ,\alpha_2\otimes t_2\big]
\,:=\,
(-1)^{\vert\alpha_1\vert}~\alpha_1 \wedge \alpha_2\otimes [t_1,t_2]\quad,
\end{flalign}
\end{subequations}
for all $\alpha_1\otimes t_1,\alpha_2\otimes t_2\in\FFF_{\mathrm{CS}} = \Omega^{\bullet+1}(M,\g)
=\Omega^{\bullet+1}(M)\otimes\g$,
where $[\,\cdot\,,\,\cdot\,]:\g\otimes\g\to \g$ is the given Lie bracket on $\g$ and
$\vert \alpha\vert = \vert\alpha\vert_{\dR}-1$ denotes the cohomological degree, which differs
from the differential form degree by $-1$. Using the Jacobi identity
for $[\,\cdot\,,\,\cdot\,]$ and the derivation property of the de Rham differential
for the $\wedge$-product (with respect to $\vert\,\cdot\,\vert_{\dR}$-degrees),
one directly checks that the homotopy Jacobi identities \eqref{eqn:homotopyJacobi} are satisfied in this example.
\sk

We would like to emphasize that the non-trivial $\ell_2$ bracket encodes in particular
the Lie bracket $\ell_2(c_1,c_2) = -[c_1,c_2]\in \FFF_{\mathrm{CS}}^{-1}$
of two infinitesimal non-Abelian gauge transformations $c_1,c_2\in \FFF_{\mathrm{CS}}^{-1}$
(the minus sign arises from our shifting conventions),
their action $A\mapsto A + \dd_\dR c + \ell_{2}(A,c) = A + \dd_\dR c + [A,c]$ 
on gauge fields $A\in \FFF_{\mathrm{CS}}^0$,
and the non-linear term of the equation of motion (i.e.\ flatness condition) $\dd_\dR A + \frac{1}{2}\,\ell_2(A,A)
= \dd_\dR A + \frac{1}{2}\,[A,A]=0$. We refer the reader to \cite[Section 4]{Jurco} for a more detailed
discussion of these important points.
\sk

For later reference, we observe that the $\ell_2$ bracket restricts
to a non-vanishing map 
\begin{flalign}\label{eqn:restrictell2}
\ell_2\vert_{\mathsf{H}^{-1}} \, : \, \mathsf{H}^{-2}\big(\Sym^2\,\FFF_{\mathrm{CS}}\big)\,=\, 
\textstyle{\bigwedge^2}\big(\mathsf{H}^{-1}\FFF_{\mathrm{CS}}\big)~
\longrightarrow~ \mathsf{H}^{-1}\FFF_{\mathrm{CS}}~~,\quad c_1\otimes c_2~\longmapsto~-[c_1,c_2]
\end{flalign}
on the degree $-1$ cohomology $\mathsf{H}^{-1}\FFF_{\mathrm{CS}}\cong \g^{\pi_0(M)}$,
i.e.\ on locally-constant $\g$-valued $0$-forms.
\end{ex}
\begin{ex}\label{ex:YM2}
The $L_\infty[1]$-algebra structure on the free Yang-Mills complex $\FFF_{\mathrm{YM}}$ from Example \ref{ex:YM}
which corresponds to non-Abelian Yang-Mills theory 
is well-known and it is presented e.g.\ in \cite[Section 5.3]{Jurco}. It consists
of non-trivial $\ell_2$ and $\ell_3$ brackets, which in particular 
encode the quadratic and cubic non-linear terms of the non-Abelian Yang-Mills equation and the 
non-Abelian features of infinitesimal gauge symmetries. 
The latter are precisely of the same form as in the Chern-Simons theory from Example \ref{ex:CS2},
i.e.\ $\ell_2(c_1,c_2) = -[c_1,c_2]\in \FFF_{\mathrm{YM}}^{-1}$
for two infinitesimal non-Abelian gauge transformations $c_1,c_2\in \FFF_{\mathrm{YM}}^{-1}$
and $A\mapsto A + \dd_\dR c + \ell_{2}(A,c) = A + \dd_\dR c + [A,c]$ 
for their action on gauge fields $A\in \FFF_{\mathrm{YM}}^0$.
We do not spell out the remaining parts of the Yang-Mills $L_\infty[1]$-algebra
structure because they will play no role in our paper.
\sk

In complete analogy to \eqref{eqn:restrictell2}, we observe that the $\ell_2$ bracket
restricts to a non-vanishing map 
$\ell_2\vert_{\mathsf{H}^{-1}} : \mathsf{H}^{-2}\big(\Sym^2\,\FFF_{\mathrm{YM}}\big) = 
\bigwedge^2 \big(\mathsf{H}^{-1}\FFF_{\mathrm{YM}})\to \mathsf{H}^{-1}\FFF_{\mathrm{YM}}$ 
on the degree $-1$ cohomology $\mathsf{H}^{-1}\FFF_{\mathrm{YM}} \cong\g^{\pi_0(M)}$.
\end{ex}

\begin{rem}
From these examples it should become evident that an $L_\infty[1]$-algebra
$(\FFF,\ell)$ provides a homological model for the solution space
of a (perturbatively interacting) physical system, which also encodes
information about the structure of its gauge symmetries. The mathematically
precise statement here is the Lurie-Pridham theorem about the equivalence of formal moduli problems
and $L_\infty$-algebras, see \cite{Lurie,Pridham}.
\end{rem}

To pass over to the dual picture of observables, one has to take
a suitable commutative dg-algebra of functions on the $L_\infty[1]$-algebra $(\FFF,\ell)$. Assuming for the moment
that $\FFF$ is a bounded cochain complex consisting of finite-dimensional vector spaces
$\FFF^i$ in each degree $i\in\bbZ$, a convenient function dg-algebra is given by the Chevalley-Eilenberg algebra
\begin{flalign}\label{eqn:CEalgebra}
\mathrm{CE}^\bullet(\FFF,\ell) \,:=\, 
\big(\Sym\,\FFF^\ast[[\lambda]],\, \delta_{\mathrm{CE}} \big)\,\in\,\dgCAlg\quad,
\end{flalign}
where $\FFF^\ast:=\hom(\FFF,\bbK)\in\Ch$ denotes the dual of the cochain complex \eqref{eqn:FFFcomplex}
and $\lambda$ is a formal parameter, interpreted as a coupling constant. The Chevalley-Eilenberg
differential $\delta_{\mathrm{CE}}$ is obtained by dualizing the $L_\infty[1]$-algebra structure
$\ell$. Dualizing the brackets $\ell_n\in \hom(\Sym^n\,\FFF,\FFF)^1$
yields degree $1$ linear maps
\begin{flalign}\label{eqn:elldual}
\ell_n^\ast  \,\in\, \hom\big(\FFF^\ast, (\Sym^n\,\FFF)^\ast\big)^1\,\cong\,\hom\big(\FFF^\ast, \Sym^n\,\FFF^\ast\big)^1\quad,
\end{flalign}
where the last isomorphism $(\Sym^n\,\FFF)^\ast\cong \Sym^n\,\FFF^\ast$ uses the finiteness assumption on 
$\FFF$ from above. Since this isomorphism is a crucial point, especially when attempting a generalization
to the infinite-dimensional case of field theories, we shall explain it in more detail. 
First, we observe that there exists the following sequence
of cochain isomorphisms
\begin{flalign}
\nn (\Sym^n\,\FFF)^\ast\,&=\, \hom\big(\Sym^n\,\FFF,\bbK\big)\,=\, \hom\big((\FFF^{\otimes n})_{\Sigma_n},\bbK\big)\\
\,&\cong\, \hom\big(\FFF^{\otimes n},\bbK\big)^{\Sigma_n}\,\cong\,\hom\big(\FFF^{\otimes n},\bbK\big)_{\Sigma_n}\quad,
\end{flalign}
which is completely general and 
does not rely on any finiteness assumptions. Steps one and two are just the definitions.
Step three uses the general result from category theory that the internal hom functor in a closed symmetric monoidal category
preserves limits, hence $\hom(-,\bbK) : \Ch^\op \to \Ch$ sends coinvariants $(-)_{\Sigma_n}$ in $\Ch$,
which are a limit in $\Ch^\op$, to invariants $(-)^{\Sigma_n}$ in $\Ch$. The last step uses
that invariants and coinvariants of finite group actions are canonically isomorphic when working over a field $\bbK$
of characteristic $0$. What is still missing, and not automatic, is to identify $\hom\big(\FFF^{\otimes n},\bbK\big)$
with $(\FFF^\ast)^{\otimes n} = \hom(\FFF,\bbK)^{\otimes n}$. There exists a canonical cochain morphism
$\hom(\FFF,\bbK)^{\otimes n} \to \hom(\FFF^{\otimes n},\bbK)$, which is defined
as the adjunct of the evaluation map $\hom(\FFF,\bbK)^{\otimes n} \otimes \FFF^{\otimes n}
\cong\big(\hom(\FFF,\bbK)\otimes\FFF\big)^{\otimes n}\stackrel{\ev^{\otimes n}}{\longrightarrow} \bbK^{\otimes n}\cong\bbK$, 
but in general this fails to be an isomorphism. The finiteness assumption that $\FFF$ is a bounded cochain 
complex consisting of finite-dimensional vector spaces implies the existence of a coevaluation
map $\mathrm{coev}: \bbK\to \FFF\otimes \FFF^{\ast}$, i.e.\ $\FFF\in\Ch$ is a rigid object, 
which is sufficient for ensuring that $(\FFF^\ast)^{\otimes n} =
\hom(\FFF,\bbK)^{\otimes n} \stackrel{\cong}{\longrightarrow} \hom(\FFF^{\otimes n},\bbK)$ is
an isomorphism. 
\sk

Post-composing the linear maps \eqref{eqn:elldual} 
with the inclusion $\Sym^n\,\FFF^\ast\hookrightarrow \Sym\,\FFF^\ast$ of the 
weight $n$ symmetric powers into the symmetric algebra allows one to define the Chevalley-Eilenberg
differential on generators $\FFF^\ast$ by 
\begin{flalign}\label{eqn:CEdifferential}
\delta_{\mathrm{CE}}\,:=\, \sum_{n=1}^{\infty}\lambda^{n-1}\, \ell_n^\ast\,\in\,\hom\big(\FFF^\ast,\Sym\,\FFF^\ast\big)^1[[\lambda]]\quad,
\end{flalign}
which is then extended to the commutative graded algebra $\Sym\,\FFF^\ast[[\lambda]]$
as a degree $1$ derivation, i.e.\ $\delta_{\mathrm{CE}}(a\,b) = \delta_{\mathrm{CE}}(a)\,b 
+ (-1)^{\vert a\vert}\,a\,\delta_{\mathrm{CE}}(b)$. The square-zero condition
$(\delta_{\mathrm{CE}})^2=0$ is equivalent to the homotopy Jacobi identities \eqref{eqn:homotopyJacobi}.
\begin{rem}\label{rem:CEclarification}
In contrast to other conventions in the literature,
our variant of the Chevalley-Eilenberg algebra given 
by \eqref{eqn:CEalgebra} and \eqref{eqn:CEdifferential} 
encodes explicitly the dependence on the formal parameter $\lambda$,
which will be very useful for organizing our perturbative expansions in this paper.
The reason for choosing the particular pattern of powers of $\lambda$ in the differential
\eqref{eqn:CEdifferential} is as follows:
The $1$-ary bracket $\ell_1=\dd$ by definition agrees with the free differential
and so must enter with no $\lambda$ dependence,
i.e.\ at order $\lambda^0 = \lambda^{1-1}$,
because the free theory is the baseline around which we perturb. We 
attach a formal parameter to the $2$-ary bracket $\lambda\,\ell_2$ 
in order to make manifest that it describes a formal perturbation. 
The pattern for the higher brackets $\{\lambda^{n-1}\ell_n\}$ is then dictated
by the homotopy Jacobi identities, or dually by the square-zero condition
$(\delta_{\mathrm{CE}})^2=0$ of the differential \eqref{eqn:CEdifferential}.
As a direct consequence of this systematic attachment of powers of $\lambda$,
the assignment of our Chevalley-Eilenberg algebras is contravariantly 
functorial $\mathrm{CE}^\bullet : \mathbf{L_\infty Alg}^\op\to \dgCAlg$,
as required for any reasonable concept of function dg-algebras.
\sk

One could further enlarge the Chevalley-Eilenberg algebra \eqref{eqn:CEalgebra}
by performing an algebraic completion of the symmetric algebra $\Sym\,\FFF^\ast[[\lambda]] 
= \bigoplus_{n=0}^\infty\Sym^n\FFF^\ast[[\lambda]]$ from polynomials to 
power series $\prod_{n=0}^\infty\Sym^n\FFF^\ast[[\lambda]]$. (It is
important to stress that this algebraic completion is unrelated to 
the topological completions arising from completed projective tensor products
of locally convex topological vector spaces in Section \ref{sec:fieldtheories} below.)
Such completions are considered for instance in the work of Costello and Gwilliam
\cite{CG1,CG2}, but they are not required, and hence will not be considered, 
in our paper because our homological constructions are well-defined on polynomials 
(due to our use of powers of $\lambda$ in the differential \eqref{eqn:CEdifferential}) 
and they admit evident extensions to power series. 
\end{rem}

We would like to emphasize that the Chevalley-Eilenberg differential \eqref{eqn:CEdifferential} decomposes
\begin{flalign}
\delta_{\mathrm{CE}} \,=\, \delta_0 + \delta\,=\, \ell_1^\ast + \sum_{n=2}^\infty \lambda^{n-1}\,\ell_n^\ast 
\end{flalign}
into an order $\lambda^0$ term $\delta_0:=\ell_1^\ast = \dd^\ast$, 
which coincides with the differential $\dd^\ast$ on the dual $\FFF^\ast$ 
of the cochain complex $\FFF$ in \eqref{eqn:FFFcomplex}, and a formal perturbation $\delta
:= \sum_{n=2}^\infty \lambda^{n-1}\,\ell_n^\ast $ which is of order $\lambda^1$ or greater. 
One should interpret the differential $\delta_0 = \dd^\ast$ as describing
the underlying free theory
of a perturbatively interacting system (obtained by sending $\lambda\to 0$) and the formal perturbation $\delta$ 
as encoding the non-linear modifications to the dynamics and gauge symmetries 
of this system. With this decomposition, the square-zero condition for the Chevalley-Eilenberg
differential is equivalent to the Maurer-Cartan equation
\begin{flalign}\label{eqn:MCequation}
(\delta_{\mathrm{CE}})^2\,=\,0
\qquad \Longleftrightarrow\qquad \partial(\delta) + \delta^2 \,=\,0\quad,
\end{flalign}
where  $\partial(\delta) = \delta_0\,\delta - (-1)^{\vert \delta\vert}\, 
\delta\,\delta_0 = \delta_0 \,\delta +\delta\,\delta_0$ denotes the `adjoint' differential
with respect to the free differential $\delta_0 =\dd^\ast$.
\sk

The construction of the Chevalley-Eilenberg algebra presented above requires some additional 
care in the context of field theories since their cochain complexes $\FFF$ consist
of smooth sections of vector bundles, hence they are infinite-dimensional.
(See, for instance, Examples \ref{ex:KG}, \ref{ex:CS} and \ref{ex:YM}.) We will address and solve these issues
in Section \ref{sec:fieldtheories} by using suitable analytical tools.


\section{\label{sec:obstructions}M{\o}ller maps for finite-dimensional systems}
In this section we introduce and analyze, from the perspective of homological algebra,
an axiomatic concept of M{\o}ller maps. To simplify the
presentation in this section, we restrict ourselves for the moment to the case of classical physical systems
which are described by an $L_\infty[1]$-algebra $(\FFF,\ell)$ whose underlying cochain complex
$\FFF$ is bounded and degree-wise finite-dimensional. In this finite-dimensional context, 
we investigate obstructions to the existence of M{\o}ller maps. The main results of this section are 
existence and also non-existence theorems for classes of physically-motivated examples.
We will extend these results in Section \ref{sec:fieldtheories} to examples of classical field theories, 
in particular to our main Examples \ref{ex:KG2}, \ref{ex:CS2} and \ref{ex:YM2}.
\sk

Let $(\FFF,\ell)$ be an $L_\infty[1]$-algebra whose underlying cochain 
complex is bounded and degree-wise finite-dimensional. As explained in Section 
\ref{sec:BVformalism}, we interpret $(\FFF,\ell)$ as a perturbatively interacting
classical physical system, whose observable dg-algebra is defined by
the Chevalley-Eilenberg algebra $\mathrm{CE}^\bullet(\FFF,\ell)\in\dgCAlg$
from \eqref{eqn:CEalgebra}. Recalling also \eqref{eqn:CEdifferential},
the Chevalley-Eilenberg differential $\delta_{\mathrm{CE}} = \delta_0 + \delta$ 
decomposes into a free part $\delta_0=\ell_1^\ast = \dd^\ast$ and an interaction part 
\begin{flalign}\label{eqn:deltapowerseries}
\delta\,:=\,\sum_{n=1}^\infty\lambda^n \,\delta_n \, :=\, \sum_{n=1}^{\infty}\,\lambda^n\,\ell_{n+1}^\ast\quad,
\end{flalign}
which are determined by dualizing the $L_\infty[1]$-algebra structure $\ell = \{\ell_n\}_{n\in\bbN^{\geq 1}}$.
Note that setting the interaction part $\delta = 0$ to zero yields the underlying free theory.
\begin{defi}\label{def:Moeller}
A {\em M{\o}ller map} is a cochain map
\begin{flalign}
K \,=\, \sum_{n=0}^\infty \lambda^n \,K_n~:~\big(\Sym\,\FFF^\ast[[\lambda]] ,\, \delta_0 \big)
~\longrightarrow~ \big(\Sym\,\FFF^\ast[[\lambda]] ,\, \delta_0 + \delta \big) \,=\, \mathrm{CE}^\bullet(\FFF,\ell)
\end{flalign}
with $K_0=\id$. An {\em algebra M{\o}ller map} is a M{\o}ller map $K$ which additionally preserves
the multiplication and unit of the symmetric algebra, i.e.\ $\mu\,(K\otimes K)=K\,\mu$
and $\eta = K\,\eta$.
\end{defi}
\begin{rem}
The zeroth-order condition $K_0=\id$ is motivated by the (trivial)
observation that the free and perturbatively interacting observables are identical at order $\lambda^0$.
From this condition it follows that $K$ is invertible 
in the sense of formal power series, which means that a M{\o}ller map (if it exists) 
defines an isomorphism between the (formal power series extension of the) cochain complex of observables
for the free system and the cochain complex of observables for the perturbatively interacting system.
\sk

Note that the isomorphism between the free and perturbatively interacting observables
given by a M{\o}ller map $K$ is not necessarily an isomorphism of commutative dg-algebras,
unless $K$ is an algebra M{\o}ller map. Conceptually, this poses no issues, since any
M{\o}ller map $K$ defines an isomorphism of commutative dg-algebras 
\begin{flalign}
K~:~\big(\Sym\,\FFF^\ast[[\lambda]] ,\, \delta_0 ,\,\mu_K^{},\,\eta_K^{}\big)~\stackrel{\cong}{\longrightarrow}~
\big(\Sym\,\FFF^\ast[[\lambda]] ,\, \delta_0 + \delta ,\, \mu,\, \eta\big)
\end{flalign}
when one endows the free observables with the
transferred multiplication and unit
\begin{flalign}
\mu_K^{}\,:=\, K^{-1}\,\mu \,(K\otimes K)\quad,\qquad
\eta_K^{} \,:=\, K^{-1}\,\eta\quad.
\end{flalign}
In other words, this means that a general (i.e.\ non-algebra) M{\o}ller map provides an 
isomorphism between the  commutative dg-algebra of observables for the perturbatively interacting system
and a commutative dg-algebra of free observables with a modified multiplication $\mu_K$ and unit $\eta_K$.
\end{rem}

It is important to emphasize that M{\o}ller maps do {\em not} always exist, even though $\delta = 
\sum_{n=1}^\infty\lambda^n\,\delta_n$ is a formal
perturbation of the free differential $\delta_0$. The existence of a M{\o}ller map is obstructed 
by a tower of successive cohomological obstructions, which all must vanish for a positive existence result.
Let us explain and clarify this statement. We start by observing that the requirement
that $K$ is a cochain map, i.e.\
\begin{flalign}\label{eqn:Moellercondition}
(\delta_0+\delta) \,K \,= \, K\,\delta_0 \qquad\Longleftrightarrow\qquad
\partial(K) + \delta\,K \,=\,0
\end{flalign}
with $\partial (K) = \delta_0\,K - K\,\delta_0$ the `adjoint' differential with 
respect to the free differential $\delta_0$,
yields upon expanding in the formal parameter $\lambda$ the tower of conditions
\begin{flalign}\label{eqn:towerofconditions}
\partial(K_n) \,=\, - \sum_{j=1}^{n}\delta_j\,K_{n-j}\quad, \quad \text{for all }n\geq 1\quad.
\end{flalign}
Let us consider the $n=1$ term of \eqref{eqn:towerofconditions}, i.e.\
\begin{flalign} \label{eqn:towerofconditions:1}
\partial(K_1)\,=\, -\delta_1\,K_0 \,=\,-\delta_1\quad.
\end{flalign}
Expanding the Maurer-Cartan equation \eqref{eqn:MCequation}
in the formal parameter $\lambda$ yields the tower of identities
\begin{subequations}\label{eqn:MCeqnexpanded}
\begin{flalign}
\partial(\delta_1)\,&=\,0\quad,\\
\partial(\delta_n)+ \sum_{j=1}^{n-1} \delta_j\,\delta_{n-j}\,&=\,0\quad,\quad\text{for all }n\geq 2\quad.
\end{flalign}
\end{subequations}
The Maurer-Cartan equation implies that the right-hand side of \eqref{eqn:towerofconditions:1} 
is $\partial$-closed, however for the existence of $K_1$ it must be exact.
Hence, there exists an obstruction to the existence of $K_1$ which is given by the cohomology class
\begin{flalign}
[\delta_1]\,\in\, \mathsf{H}^1\hom\big((\Sym\,\FFF^\ast,\delta_0),\,(\Sym\,\FFF^\ast,\delta_0)\big)\quad,
\end{flalign}
i.e.\ $K_1$ exists if and only if $[\delta_1]=0$ is trivial.
\sk

This pattern extends to higher orders in $\lambda$ and it yields a tower of 
successive obstruction classes for the existence of $K$. Suppose that
a M{\o}ller map $K$ exists up to order $\lambda^{n-1}$, for some $n\geq 2$,
and consider the $n$-th term of \eqref{eqn:towerofconditions}.
Applying the differential $\partial$ to the right-hand side of \eqref{eqn:towerofconditions}, 
we compute 
\begin{flalign}
\nn \partial\bigg(- \sum_{j=1}^{n}\delta_j\,K_{n-j}\bigg) \,&=\, - \sum_{j=1}^n \partial(\delta_j)\,K_{n-j}
+ \sum_{j=1}^{n}\delta_j\,\partial(K_{n-j})\\
\,&=\,  \sum_{j=2}^n \sum_{k=1}^{j-1}\delta_k\,\delta_{j-k}\,K_{n-j}-\sum_{j=1}^{n-1}\sum_{k=1}^{n-j}\delta_j\,\delta_k\,K_{n-j-k}\,=\,0\quad,
\end{flalign}
where in the second step we used the order $<n$ terms of \eqref{eqn:towerofconditions} 
and the Maurer-Cartan equation \eqref{eqn:MCeqnexpanded}.
Hence, the right-hand side of \eqref{eqn:towerofconditions} defines a class
in $\mathsf{H}^1\hom\big((\Sym\,\FFF^\ast,\delta_0),(\Sym\,\FFF^\ast,\delta_0)\big)$
and a solution of \eqref{eqn:towerofconditions} for $K_n$ can be found if and only if this cohomology class is trivial.
Let us summarize this simple observation in the following
\begin{propo}\label{propo:obstructions}
A M{\o}ller map $K = \sum_{n=0}^\infty\lambda^n\,K_n$
exists if and only if the tower of successive obstructions given by the cohomology classes
\begin{flalign}
\bigg[- \sum_{j=1}^{n}\delta_j\,K_{n-j}\bigg]\,\in\, 
\mathsf{H}^1\hom\big((\Sym\,\FFF^\ast,\delta_0),\,(\Sym\,\FFF^\ast,\delta_0)\big)
\quad,\quad\text{for all }n\geq 1\quad,
\end{flalign}
is trivial.
\end{propo}
\begin{rem}
Note that these obstructions are determined by the interaction term $\delta = \sum_{n=1}^\infty \lambda^n\,\delta_n$ but they live in a cohomology group which is associated with the free theory $(\Sym\,\FFF^\ast,\delta_0)$. 
The existence or non-existence of M{\o}ller maps will thus depend very sensitively 
on specific details of both the free theory and the perturbative interaction.
\end{rem}

We will now prove an existence theorem for M{\o}ller maps under hypotheses
which are inspired by the interacting Klein-Gordon field from Examples \ref{ex:KG}
and \ref{ex:KG2}. For this we observe that a crucial feature of this example 
is that the complex $\FFF$ is concentrated only in degrees $\{0,1\}$ and that 
the degree $1$ cohomology $\mathsf{H}^1\FFF = 0$ is trivial.
\begin{theo}\label{theo:existence}
Suppose that $(\FFF,\ell)$ is an $L_\infty[1]$-algebra whose underlying cochain complex
\begin{flalign}
\FFF\,=\,\Big(\xymatrix{
\FFF^0\ar[r]^-{\dd}\,&\,\FFF^1
}
\Big)\,\in\,\Ch
\end{flalign}
is concentrated in degrees $\{0,1\}$, with $\FFF^0$ and $\FFF^1$ finite-dimensional vector spaces.
If the cohomology $\mathsf{H}^1\FFF=0$ is trivial, 
then there exists an algebra M{\o}ller map, defined explicitly by \eqref{eqn:Kscalar}.
\end{theo}
\begin{proof}
Note that the dual complex
\begin{flalign}
\FFF^\ast \,=\,\Big(\xymatrix{
\stackrel{(-1)}{(\FFF^1)^\ast} \ar[r]^-{\delta_0 \,=\, \dd^\ast }\,&\,\stackrel{(0)}{(\FFF^0)^\ast}
}
\Big)\,\in\,\Ch
\end{flalign}
is concentrated in degrees $\{-1,0\}$. 
From the hypothesis that $\mathsf{H}^{1}\FFF = 0$ is trivial, it follows
that the degree $-1$ cohomology $\mathsf{H}^{-1}\FFF^\ast=0$ of the dual complex is trivial too.
(Use exactness of the internal hom functor $\hom(-,\bbK)$, for $\bbK$ a field.)
It follows that the sequence
\begin{flalign}
\xymatrix{
0 \ar[r] \,&\, (\FFF^{1})^\ast \ar[r]^-{\delta_0} \,&\, (\FFF^0)^\ast \ar[r] \,&\,  (\FFF^{0})^\ast/\delta_0  (\FFF^{1})^\ast\ar[r]\,&\, 0
}
\end{flalign}
is short exact. Using again that $\bbK$ is a field, 
there exists a splitting of this short exact sequence, i.e.\ a linear map 
$s : (\FFF^0)^\ast  \to (\FFF^1)^\ast $ satisfying $s\,\delta_0 =\id$. 
\sk

We propose a candidate for an algebra M{\o}ller map $K$ by setting on generators
\begin{flalign}\label{eqn:Kscalar}
K(\varphi) \,:=\ \varphi + \delta s(\varphi)\quad,\qquad 
K(\varphi^\ddagger)\,:=\,\varphi^\ddagger\quad,
\end{flalign}
for all $\varphi\in (\FFF^0)^\ast$ and $\varphi^\ddagger\in(\FFF^1)^\ast $,
and extending $K$ to $\Sym\,\FFF^\ast[[\lambda]]$ as a morphism of commutative graded algebras.
By construction as a graded commutative algebra morphism, 
it suffices to check the condition \eqref{eqn:Moellercondition} on generators
in order to prove that $K$ defines indeed an algebra M{\o}ller map.
For generators $\varphi\in (\FFF^0)^\ast$ in degree $0$, we compute
\begin{flalign}
\delta_0 K(\varphi) + \delta K(\varphi) - K(\delta_0\varphi) \,=\, 0+0-0 \,=\, 0\quad, 
\end{flalign}
where we used that both $\delta_0$ and $\delta$ act trivially on degree $0$ elements for degree reasons.
For generators $\varphi^\ddagger\in(\FFF^{1})^\ast$ in degree $-1$, we compute
\begin{flalign}
\delta_0 K(\varphi^\ddagger) + \delta K(\varphi^\ddagger) - K(\delta_0\varphi^\ddagger) \,=\,
\delta_0\varphi^\ddagger + \delta \varphi^\ddagger - \delta_0\varphi^\ddagger -\delta s (\delta_0\varphi^\ddagger)
\,=\,\delta \varphi^\ddagger -\delta\varphi^\ddagger \,=\,0\quad,
\end{flalign}
where in the second step we used the splitting property $s\,\delta_0=\id$. 
This shows that $K$ defines an algebra M{\o}ller map.
\end{proof}

We will now prove a \textit{non-existence} theorem for M{\o}ller maps under hypotheses
which are inspired by the non-Abelian Chern-Simons and Yang-Mills fields from Examples \ref{ex:CS}, \ref{ex:YM},
\ref{ex:CS2} and \ref{ex:YM2}. For this we observe that a crucial feature of
these examples is that the complex $\FFF$ is concentrated in degrees $\geq -1$
and that the restriction $\ell_2\vert_{\mathsf{H}^{-1}} : 
\mathsf{H}^{-2}\big(\Sym^2\,\FFF\big) = \bigwedge^2\big(\mathsf{H}^{-1}\FFF\big)\to \mathsf{H}^{-1}\FFF$ 
of the $\ell_2$ bracket to the degree $-1$ cohomology is non-vanishing.
\begin{theo}\label{theo:nonexistence}
Suppose that $(\FFF,\ell)$ is an $L_\infty[1]$-algebra whose underlying
cochain complex
\begin{flalign}
\FFF\,=\,\Big(
\xymatrix{
\FFF^{-1}\ar[r]^-{\dd}\,&\, \FFF^0 \ar[r]^-{\dd}\,&\, \cdots  \ar[r]^-{\dd}\,&\, \FFF^N
}
\Big)\,\in\,\Ch
\end{flalign}
is concentrated in degrees $\{-1,0,\dots,N\}$ for some $-1\leq N\in\bbZ$, with all $\FFF^i$ finite-dimensional vector spaces.
If the restriction $\ell_2\vert_{\mathsf{H}^{-1}} : 
\mathsf{H}^{-2}\big(\Sym^2\,\FFF\big) = \bigwedge^2\big(\mathsf{H}^{-1}\FFF\big)\to \mathsf{H}^{-1}\FFF$  
of the $\ell_2$ bracket
to the cohomology $\mathsf{H}^{-1}\FFF$ is non-vanishing, then there does not exist a M{\o}ller map.
\end{theo}
\begin{proof}
Our proof strategy is to use homological perturbation theory \cite{HPT}
in order to show that the cochain complexes
$\big( \Sym\,\FFF^\ast[[\lambda]],\,\delta_0\big)$ and  $\big( \Sym\,\FFF^\ast[[\lambda]],\,\delta_0 + \delta\big)$
can not be quasi-isomorphic under the hypotheses of this theorem. 
It then follows immediately that there does not exist a M{\o}ller map in this case.
\sk

Let us start by observing that the dual complex
\begin{flalign}
\FFF^\ast \,=\,\Big(\xymatrix{
\stackrel{(-N)}{(\FFF^N)^\ast} \ar[r]^-{\delta_0 \,=\,\dd^\ast }\,&\,\cdots
\ar[r]^-{\delta_0 \,=\,\dd^\ast }
\,&\,\stackrel{(0)}{(\FFF^0)^\ast} \ar[r]^-{\delta_0 \,=\,\dd^\ast }\,&\, \stackrel{(1)}{(\FFF^{-1})^\ast}
}
\Big)\,\in\,\Ch
\end{flalign}
is concentrated in degrees $\{-N,\dots,0,1\}$. 
Since $\bbK$ is a field, there exists a deformation retract
\begin{equation}\label{eqn:generatorsSDR}
\begin{tikzcd}
(\mathsf{H}\FFF^\ast,0) \ar[r,shift right=-1ex,"i"] & \ar[l,shift right=-1ex,"p"] (\FFF^\ast,\delta_0) \ar[loop,out=-20,in=20,distance=20,swap,"h"]
\end{tikzcd}
\end{equation}
of the dual complex $(\FFF^\ast,\delta_0)$ onto its cohomology $(\mathsf{H}\FFF^\ast,0)$,
which we consider as a cochain complex with trivial differential. In top degree $1$,
the linear map $p : (\FFF^\ast)^1 = (\FFF^{-1})^\ast  \to \mathsf{H}^1\FFF^\ast = (\FFF^{-1})^\ast/\delta_0(\FFF^{0})^\ast$
is simply the quotient map and $i : \mathsf{H}^1\FFF^\ast = (\FFF^{-1})^\ast/\delta_0(\FFF^{0})^\ast 
\to (\FFF^\ast)^1 =(\FFF^{-1})^\ast $ picks representatives for the equivalence classes.
\sk

The deformation retract \eqref{eqn:generatorsSDR} can be lifted along $\Sym$ and extended
to formal power series, which yields a deformation retract
\begin{equation}\label{eqn:freeSDR}
\begin{tikzcd}
\big(\Sym\,\mathsf{H}\FFF^\ast[[\lambda]],\, 0\big) \ar[r,shift right=-1ex,"I"] & \ar[l,shift right=-1ex,"P"]
\big( \Sym\,\FFF^\ast[[\lambda]],\,\delta_0\big) \ar[loop,out=-15,in=15,distance=25,swap,"H"]
\end{tikzcd}\quad,
\end{equation}
see e.g.\ \cite[Proposition 2.5.5]{Gwilliam} or \cite[Section 4.1]{Bruinsma} for the 
details of this construction.
Applying homological perturbation theory to the formal perturbation $\delta = \sum_{n= 1}^\infty\lambda^n\,\delta_n$ 
yields a perturbed deformation retract
\begin{equation}\label{eqn:interactingSDR}
\begin{tikzcd}
\big(\Sym\,\mathsf{H}\FFF^\ast[[\lambda]],\, \widetilde{\delta}\,\big) \ar[r,shift right=-1ex,"\widetilde{I}"] 
& \ar[l,shift right=-1ex,"\widetilde{P}"]\big( \Sym\,\FFF^\ast[[\lambda]], \,\delta_0 + \delta\big) \ar[loop,out=-10,in=10,distance=22,swap,"\widetilde{H}"]
\end{tikzcd}\quad.
\end{equation}
Explicit expressions for $\widetilde{\delta}, \widetilde{I}, \widetilde{P}, \widetilde{H}$ 
in terms of $\delta, I, P, H$ can be found in \cite{HPT}.
\sk

Combining the two deformation retracts \eqref{eqn:freeSDR} and \eqref{eqn:interactingSDR},
we arrive at the following observation:
If there exists a M{\o}ller map
$K: \big(\Sym\,\FFF^\ast[[\lambda]],\,\delta_0\big)\to \big(\Sym\,\FFF^\ast[[\lambda]],\,\delta_0+\delta\big) $,
then $\widetilde{P}\,K\,I : \big(\Sym\,\mathsf{H}\FFF^\ast[[\lambda]],\,0\big)\to 
\big(\Sym\,\mathsf{H}\FFF^\ast[[\lambda]],\,\widetilde{\delta}\,\big)$ defines an isomorphism of cochain
complexes. 
(To check this claim, note that the order $\lambda^0$ of $\widetilde{P}\,K\,I$ 
is $P\,I=\id$.) 
This would imply that $\widetilde{\delta}=0$. Hence, one can prove
that there does {\em not} exist a M{\o}ller map by showing that 
$\widetilde{\delta} \neq 0$.
\sk

To conclude the proof, we recall from \cite{HPT} the explicit formula
\begin{flalign}
\widetilde{\delta} \,=\, P\, \big(\id - \delta\,H\big)^{-1}\,\delta\,I
\,=\, P\,\Big(\sum_{n=0}^\infty (\delta\,H)^n \Big)\,\delta\,I = \lambda\, P\,\delta_1\,I \,+\, \mathcal{O}(\lambda^2)
\end{flalign}
for the perturbed differential. Restricting the leading term $\widetilde{\delta}_1 := P\,\delta_1\,I
= P\,\ell_2^\ast\,I$ to degree $1$ generators $\mathsf{H}^1\FFF^\ast$ 
yields the linear map
\begin{flalign}
\widetilde{\delta}_1\vert_{\mathsf{H^1}} \,:\,\xymatrix@C=2em{
\mathsf{H}^1\FFF^\ast \ar[r]^-{i} \,&\,(\FFF^\ast)^1 \ar[r]^-{\ell_2^\ast}\,&\, 
(\Sym^2\,\FFF^{\ast})^2\,=\,\textstyle{\bigwedge^2}\,(\FFF^{\ast})^1 
\ar[r]^-{p\wedge p}\,&\,\textstyle{\bigwedge^2}\,\big(\mathsf{H}^1\FFF^{\ast}\big)
}\quad,
\end{flalign}
where $i$ and $p$ are from the deformation retract for generators \eqref{eqn:generatorsSDR}.
Recalling from above that, in top degree $1$, $p$ is the quotient map to cohomology 
$\mathsf{H}^1\FFF^\ast = (\FFF^{-1})^\ast/\delta_0(\FFF^{0})^\ast$
and $i$ picks representatives of equivalence classes, 
it follows that $\widetilde{\delta}_1\vert_{\mathsf{H^1}}$
coincides with the linear map which is obtained by descending $\ell_2^\ast$ to the quotient
$\mathsf{H}^1\FFF^\ast = (\FFF^{-1})^\ast/\delta_0(\FFF^{0})^\ast$.
This is the dual of the linear map $\ell_2\vert_{\mathsf{H}^{-1}}$,
which by our hypothesis is non-vanishing. It then 
follows that $\widetilde{\delta}_1\vert_{\mathsf{H^1}}\neq 0$
and consequently $\widetilde{\delta} \neq 0$.
Therefore, M{\o}ller maps do not exist under the hypotheses of this theorem.
\end{proof}
\begin{rem}\label{rem:origin}
Recalling the field theoretic context from Examples \ref{ex:CS}, \ref{ex:YM},
\ref{ex:CS2} and \ref{ex:YM2}, we can provide an interpretation for the origin of
this non-existence theorem for M{\o}ller maps: The degree $-1$ cohomology $\mathsf{H}^{-1}\FFF$
is non-trivial if and only if the linearized gauge transformations of the theory
have non-trivial stabilizers. This is a generic feature of gauge theories treated
perturbatively around the trivial background solution, but it may be altered
when perturbing around non-trivial backgrounds, see Remark \ref{rem:background}.
The restricted linear map $\ell_2\vert_{\mathsf{H}^{-1}}$ is non-vanishing
if and only if the Lie bracket for infinitesimal non-Abelian gauge transformations
operates non-trivially on such stabilizers. This is, in particular, the case for all gauge theories
that are modeled by non-Abelian connections and their infinitesimal non-Abelian gauge transformations,
independently of their specific dynamics. Summing up, the non-existence result
for M{\o}ller maps in Theorem \ref{theo:nonexistence} has a gauge theoretic origin, and {\em not} a dynamical one. 
\end{rem}


\section{\label{sec:fieldtheories}M{\o}ller maps for classical field theories}
In this section we focus on genuine classical field theoretic examples.
We prove existence and non-existence results for M{\o}ller maps 
in this context that parallel the simpler results 
(Theorems \ref{theo:existence} and \ref{theo:nonexistence}) 
for finite-dimensional systems. 
A framework that naturally suits the current field theoretic context is that 
of cochain complexes of complete locally convex topological vector spaces. 
Unfortunately, it is well-known that topological vector spaces 
of various flavors typically fail to form an Abelian category, 
which leads to difficulties in defining cohomologies 
and ultimately constitutes the main challenge in the 
development of a well-behaved homological algebra 
in such functional analytic contexts. Although recent approaches 
to overcome this difficulty exist, 
e.g.\ based on exact categories \cite{Kelly} 
or condensed mathematics \cite{condensed}, 
here we prefer to adopt the following more elementary 
approach: Instead of arguing abstractly as in Theorems \ref{theo:existence} 
and \ref{theo:nonexistence} by using the concepts of
cohomologies and quasi-isomorphisms, we will avoid forming cohomologies
and take a more constructive, and hence computationally
harder, approach in which cochain complexes
of complete locally convex topological vector spaces are compared through 
continuous deformation retracts.
This provides concrete data witnessing the equivalence
of two complexes which is insensitive to the 
specific details of the ambient category. In particular, 
our concrete constructions of continuous deformation retracts
yield weak equivalences when embedded into the recent frameworks \cite{Kelly,condensed}
for functional analytic homological algebra.
\sk

To capture field theory models, we now drop the requirement that the cochain complex $\FFF$
from \eqref{eqn:FFFcomplex} consists of finite-dimensional vector spaces $\FFF^i$.
Instead, we shall assume that $\FFF$ is bounded, its non-trivial 
components $\FFF^i = \Gamma^\infty(M,F^i)$ are nuclear Fr\'echet spaces of 
smooth sections of vector bundles $F^i \to M$ over a manifold $M$,
and that the differential $\dd$ of $\FFF$ consists of (continuous) linear differential operators. 
Indeed, Examples \ref{ex:KG}, \ref{ex:CS} and \ref{ex:YM} are precisely of this form. 
Furthermore, we assume that the perturbative interactions are encoded 
through an $L_\infty[1]$-algebra structure $\ell =\{\ell_n\}_{n\in\bbN^{\geq 1}}$
whose components $\ell_n$ are continuous. 
One realizes immediately that Examples 
\ref{ex:KG2}, \ref{ex:CS2} and \ref{ex:YM2} meet the additional continuity requirement.
We shall use the convenient terminology of \textit{field theory $L_\infty[1]$-algebras} 
for such analytical variants of $L_\infty[1]$-algebras $(\FFF,\ell)$. 
\sk 

In preparation to passing over to observables, let us recall that 
the category of complete locally convex topological vector spaces 
and continuous linear maps forms a symmetric monoidal category with 
respect to the completed projective tensor product $\widehat{\otimes}$, see e.g.\ \cite[Section 15]{Jarchow}.
This symmetric monoidal category is not closed, i.e.\ there does not exist an internal hom,
which however will cause no issues in our construction below since $\FFF^i = \Gamma^\infty(M,F^i)$ 
are well-behaved objects. In analogy with \eqref{eqn:otimes}, one endows 
the category of cochain complexes of complete locally convex 
topological vector spaces with the completed projective tensor product, which we 
denote again by $\widehat{\otimes}$. 
Note that, due to continuity of the differentials involved 
and density of the algebraic tensor product $\otimes$ 
in the completed projective tensor product $\widehat{\otimes}$, 
the Leibniz rule completely determines the differential 
on the tensor product complex. 
A commutative, associative and unital algebra 
$A = (A,\mu,\eta)$ in the symmetric monoidal category 
of cochain complexes of complete locally convex topological vector spaces 
is then just a commutative dg-algebra whose components 
are furthermore endowed with complete locally convex vector space topologies 
and whose differential $\dd$, multiplication $\mu$ and unit $\eta$ 
are degree-wise continuous.
\sk

The construction of the
free commutative algebra over a cochain complex of complete locally convex topological vector spaces $V$
is in general difficult since, even though $\widehat{\otimes}$ preserves coinvariants (which are a quotient)
\cite[Proposition 15.2.1]{Jarchow}, it does {\it not} in general preserve direct sums. 
These difficulties disappear when all components $V^i$ of $V$ are sufficiently
well-behaved objects. One class of such well-behaved objects is 
for instance that of DF-spaces, see e.g.\ \cite[Section 12.4]{Jarchow}. 
For our purposes, the niceness of DF-spaces lies in the fact that 
forming the completed projective tensor product with a DF-space 
preserves countable direct sums, see e.g.\ \cite[Theorem 15.5.3]{Jarchow}, 
which allows us to define the multiplication of the completed 
symmetric algebra in \eqref{eqn:completedSym} as usual. 
Examples of DF-spaces are obtained as strong duals of Fr\'echet spaces \cite[Theorem 12.4.5]{Jarchow}. 
For instance, the vector spaces of compactly supported distributional 
sections of a vector bundle endowed with the usual strong topology 
are DF-spaces. 
Indeed, the DF-spaces encountered later on are precisely of this form, 
see \eqref{eqn:completedCEalgebra}. 
\sk

Given any cochain complex of complete DF-spaces $V$,
a model for the free commutative algebra over $V$ is given by the {\it completed} symmetric algebra 
\begin{flalign}\label{eqn:completedSym}
\widehat{\Sym}\,V\,:=\, \bigoplus_{n=0}^\infty {\widehat{\Sym}}^n\,V\,:=\, 
\bigoplus_{n=0}^\infty \big(V^{\widehat{\otimes}n}\big)_{\Sigma_n}\quad,
\end{flalign}
which is defined as in \eqref{eqn:Sym} by replacing the algebraic tensor product $\otimes$ 
with the completed projective tensor product $\widehat{\otimes}$. (Note that the components
of $\big(V^{\widehat{\otimes}n}\big)_{\Sigma_n}$ are DF-spaces since this class of
spaces is stable under countable direct sums, quotients and completions \cite[Section 29.5]{Koethe}
and under the completed projective tensor product \cite[Theorem 15.6.2]{Jarchow}.)
\begin{rem}\label{rem:completedSym}
Note that the cochain complex of complete locally convex topological 
vector spaces underlying \eqref{eqn:completedSym}
can be defined for \textit{any} cochain complex 
of complete locally convex topological vector spaces (not necessarily DF-spaces) $V = (V,\dd)$,
but this does not in general carry a commutative dg-algebra structure unless $V$ is sufficiently nice,
e.g.\ all components are DF-spaces. The relevant construction is given by
defining the underlying complete locally convex topological graded vector space $\widehat{\Sym}\, V$ 
according to \eqref{eqn:completedSym} and endowing it with the differential given by
\begin{flalign}
\begin{gathered}
\xymatrix{
\widehat{\Sym}\, V \ar@{-->}[rrrrrr]^-{\dd} \,&&&&&&\, \widehat{\Sym}\, V \\ 
\big( V^{\widehat{\otimes} n} \big)_{\Sigma_n} \ar[u]^-{\iota_n} \ar[r]_-{a_n} \,&\, V^{\widehat{\otimes} n} \ar[rrrr]_-{\sum_{k=0}^{n-1} \id^{\widehat{\otimes} k} \widehat{\otimes} \dd \widehat{\otimes} \id^{\widehat{\otimes} n-k-1}} \,&&&&\, V^{\widehat{\otimes} n} \ar[r]_-{q_n} \,&\, \big( V^{\widehat{\otimes} n} \big)_{\Sigma_n} \ar[u]_-{\iota_n} 
}
\end{gathered}
\quad.
\end{flalign} 
Here $\iota_n$ denotes the inclusion of the $n$-th direct summand, 
$a_n$ the group-averaging map and $q_n$ the quotient map, 
all of which are manifestly continuous linear maps. 
The differential $\dd$ on $\widehat{\Sym}\, V$ squares to zero 
as a consequence of the explicit form of the intermediate step in the bottom row 
and the fact that the original differential on $V$ squares to zero. 
\end{rem}

With the above preparations,
we can pass over from a field theory $L_\infty[1]$-algebra $(\FFF,\ell)$ to its 
observables by forming the {\it completed} Chevalley-Eilenberg algebra 
\begin{flalign}\label{eqn:completedCEalgebra}
\widehat{\mathrm{CE}}^\bullet(\FFF,\ell) \,:=\, 
\big(\widehat{\Sym}\,\FFF^\prime[[\lambda]],\, \delta_{\mathrm{CE}} \big)\quad,
\end{flalign}
where the term completed refers here to a topological completion, 
arising from the completed projective tensor products in \eqref{eqn:completedSym},
and not to the unrelated concept of an algebraic completion from polynomials
to power series, see also Remark \ref{rem:CEclarification}.
As in the finite-dimensional case \eqref{eqn:CEalgebra}, 
$\lambda$ is a formal parameter interpreted as a coupling constant.
Here $\FFF^\prime$ denotes the cochain complex 
of complete DF-spaces whose components are the strong dual spaces $(\FFF^\prime)^i := (\FFF^{-i})^\prime$, 
for all $i \in \bbZ$, 
i.e.\ the vector spaces of {\it continuous} linear functionals on the 
nuclear Fr\'echet spaces $\FFF^{-i} = \Gamma^{\infty}(M,F^{-i})$ endowed with the strong dual topology. 
(Note that, by construction, $(\FFF^\prime)^i$ is the familiar 
space of compactly supported distributional sections of the dual of the vector bundle $F^{-i}\to M$.) 
The differential $\dd^\prime$ of $\FFF^\prime$ is the transpose of the differential $\dd$ of $\FFF$. 
The Chevalley-Eilenberg differential $\delta_{\mathrm{CE}}$ 
is obtained by transposing the continuous $L_\infty[1]$-algebra structure $\ell$ of $\FFF$, as explained below.
Consider the degree $1$ continuous linear maps $\ell_n: \Sym^n\,\FFF \to \FFF$. 
Because $\FFF$ is complete, these can be equivalently seen as 
continuous linear maps $\ell_n: \widehat{\Sym}^n\,\FFF \to \FFF$ 
out of the weight $n$ completed symmetric powers.
Transposition yields degree $1$ continuous linear maps
\begin{flalign}\label{eqn:elltranspose}
\ell_n^\prime\,:\, \FFF^\prime ~\longrightarrow~ 
 \big(\widehat{\Sym}^n\,\FFF\big)^\prime\,\cong\, \widehat{\Sym}^n\,\FFF^\prime\quad,
\end{flalign}
where the last topological cochain isomorphism arises as follows. 
First, we observe that there exists the following sequence
of topological cochain isomorphisms
\begin{flalign}
\big(\widehat{\Sym}^n\,\FFF\big)^\prime\,=\, \big((\FFF^{\widehat{\otimes} n})_{\Sigma_n}\big)^\prime\,\cong\, \big((\FFF^{\widehat{\otimes} n})^\prime\big)^{\Sigma_n}\,\cong\,\big((\FFF^{\widehat{\otimes} n})^\prime\big)_{\Sigma_n}\quad.
\end{flalign}
Step one is just the definition. Step two is the observation that 
continuous linear functionals on coinvariants are simply 
invariant continuous linear functionals and furthermore that 
the strong dual topology on the left-hand side 
agrees with the subspace topology on the right-hand side. (This follows 
from continuity of the group-averaging construction
associated with the continuous linear action of a finite group.)
The last step uses that invariants and coinvariants 
for the continuous linear action of a finite group 
are topologically linear isomorphic. 
(This follows once again from continuity of the above-mentioned group-averaging construction.) 
It remains to identify $\big(\FFF^{\widehat{\otimes} n}\big)^\prime$
with $(\FFF^\prime)^{\widehat{\otimes} n}$. For this purpose, 
recall the following facts: (1)~the components of $\FFF$ are nuclear Fr\'echet spaces; 
(2)~nuclear Fr\'echet spaces are closed under finite direct sums 
(which appear as the components of the completed tensor powers of the bounded complex $\FFF$), 
see \cite[Section 18.3]{Koethe} and \cite[Corollary 21.2.3]{Jarchow}; 
(3)~strong duals commute with finite direct sums 
(because finite direct sums coincide with finite direct products, 
see e.g.\ \cite[Section 18.5]{Koethe}). 
With these facts in mind, the topological cochain isomorphism 
$\big(\FFF^{\widehat{\otimes} n}\big)^\prime \cong (\FFF^\prime)^{\widehat{\otimes} n}$ 
is an immediate consequence of the following result due to Grothendieck, 
see e.g.\ \cite[Theorem 21.5.9]{Jarchow}: 
When $X$ is a nuclear Fr\'echet space and $Y$ is a Fr\'echet space, 
the strong dual space $(X \widehat{\otimes} Y)^\prime$ of the completed 
projective tensor product is (canonically) topologically linear isomorphic 
to the completed projective tensor product $X^\prime \widehat{\otimes} Y^\prime$ 
of the strong dual spaces. 
\sk

Post-composing the continuous linear maps \eqref{eqn:elltranspose} 
with the inclusion $\widehat{\Sym}^n\,\FFF^\prime\hookrightarrow \widehat{\Sym}\,\FFF^\prime$ of the 
weight $n$ completed symmetric powers into the completed symmetric algebra 
allows one to define the Chevalley-Eilenberg
differential on generators $\FFF^\prime$ as the degree 1 continuous linear map
\begin{flalign}\label{eqn:CEdifferential2}
\delta_{\mathrm{CE}}\,:=\,\delta_0 + \delta \,:=\, \ell_1^\prime + \sum_{n=2}^{\infty}\lambda^{n-1}\, \ell_n^\prime\,:\, \FFF^\prime ~\longrightarrow ~\widehat{\Sym}\,\FFF^\prime[[\lambda]]\quad,
\end{flalign}
which is then extended to the complete locally convex topological 
commutative graded algebra $\widehat{\Sym}\,\FFF^\prime[[\lambda]]$
as a degree $1$ continuous derivation, i.e.\ $\delta_{\mathrm{CE}}(a\,b) = \delta_{\mathrm{CE}}(a)\,b 
+ (-1)^{\vert a\vert}\,a\,\delta_{\mathrm{CE}}(b)$. The square-zero condition
$(\delta_{\mathrm{CE}})^2=0$ is equivalent to the homotopy Jacobi identities \eqref{eqn:homotopyJacobi}.
\begin{defi}\label{def:Moeller-cont}
A {\em continuous M{\o}ller map} is a continuous cochain map
\begin{flalign}
K \,=\, \sum_{n=0}^\infty \lambda^n \,K_n~:~\big(\widehat{\Sym}\,\FFF^\prime[[\lambda]] ,\, \delta_0 \big)
~\longrightarrow~ \big(\widehat{\Sym}\,\FFF^\prime[[\lambda]] ,\, \delta_0 + \delta \big) \,=\, \widehat{\mathrm{CE}}^\bullet(\FFF,\ell)
\end{flalign}
with $K_0=\id$. A {\em continuous algebra M{\o}ller map} is a continuous 
M{\o}ller map $K$ which additionally preserves
the multiplication and unit of the completed symmetric algebra, i.e.\ $\mu\,(K\widehat{\otimes} K)=K\,\mu$
and $\eta = K\,\eta$.
\end{defi}

We will now prove an existence theorem for continuous M{\o}ller maps 
for a class of field theories that includes the interacting Klein-Gordon field 
from Examples \ref{ex:KG} and \ref{ex:KG2}. 
As in the related Theorem \ref{theo:existence} in the finite-dimensional setting, 
we assume that the underlying complex of the field theory $L_\infty[1]$-algebra 
$(\FFF,\ell)$ is concentrated in degrees $\{0,1\}$, thereby excluding non-trivial gauge symmetries.
As explained at the beginning of this section, 
our current functional analytic setup requires us to refine the
assumption on the vanishing of the first cohomology from Theorem \ref{theo:existence} 
to the more constructive datum of a continuous splitting of
a sequence of complete locally convex topological vector spaces. 
\begin{theo}\label{theo:FTexistence}
Suppose that $(\FFF,\ell)$ is a field theory $L_\infty[1]$-algebra whose 
underlying cochain complex
\begin{flalign}
\FFF\,=\,\Big(\xymatrix{
\FFF^0\ar[r]^-{\dd}\,&\,\FFF^1
}
\Big)
\end{flalign}
is concentrated in degrees $\{0,1\}$. 
Assume that a continuous splitting $r: \FFF^1 \to \FFF^0$ of the sequence 
\begin{flalign}\label{eqn:FTsequence}
\xymatrix{
0\ar[r]\,&\,\ker\dd\ar[r]^-{\subseteq}\,&\,\FFF^0\ar[r]^-{\dd}\,&\,\FFF^1\ar[r]\,&\,0
}
\end{flalign}
of complete locally convex topological vector spaces 
exists, namely such that $\dd\, r = \id$. Then there exists a continuous algebra 
M{\o}ller map, defined explicitly by \eqref{eqn:Kscalar} with $s=r^\prime$ the transpose.
\end{theo}
\begin{proof}
The transpose $s=r^\prime: (\FFF^0)^\prime \to (\FFF^1)^\prime$ of $r$ 
yields a continuous splitting of the continuous dual of the sequence 
\eqref{eqn:FTsequence}, namely such that $s\,\delta_0 = \id$ with $\delta_0 = \dd^\prime$. 
We define a candidate continuous algebra M{\o}ller map $K$ 
on generators by the same formula \eqref{eqn:Kscalar} as in the finite-dimensional case. 
Since $\delta$ and $s$ are both continuous, the resulting linear map 
is continuous too and hence it admits a unique extension $K$ 
to the completed symmetric algebra. By construction, $K$ is compatible 
with multiplications and units. To verify that $K$ intertwines between the differentials 
$\delta_0$ and $\delta_{\mathrm{CE}}=\delta_0+\delta$, it suffices to consider generators. 
For those the relevant checks are identical to the ones performed in the proof of Theorem \ref{theo:existence}. 
\end{proof}

\begin{ex}\label{ex:KGexistence}
For the interacting Klein-Gordon field
from Examples \ref{ex:KG} and \ref{ex:KG2},
the continuous splitting required
by Theorem \ref{theo:FTexistence} is a continuous linear map
$r : C^\infty(M)\to C^\infty(M)$ that assigns solutions to the inhomogeneous
Klein-Gordon equation $(\square +m^2)r(\Phi^\ddagger) = \Phi^\ddagger$,
for all $\Phi^\ddagger \in C^\infty(M)$. There exist different options to construct
such splittings, which are all rooted in the fact that $\square +m^2$
is a normally hyperbolic (or, more generally, a Green hyperbolic) 
linear differential operator on a globally hyperbolic Lorentzian manifold $M$.
These choices lead via Theorem \ref{theo:FTexistence} to different continuous algebra M{\o}ller maps.
\sk

As a first option, we can use well-posedness of the inhomogeneous initial value problem for any 
spacelike Cauchy surface $j:\Sigma \hookrightarrow M$. 
Consider the continuous linear map 
\begin{flalign}
\nn \mathsf{data}\,:\,C^\infty(M) \,&\longrightarrow\, C^\infty(M) \times C^\infty(\Sigma) \times \Omega^{m-1}(\Sigma)\quad, \\
\Phi \,&\longmapsto\, \left( (\square+m^2) \Phi,\, j^\ast (\Phi),\, j^\ast(\ast \dd_{\dR} \Phi) \right)
\end{flalign}
between Fr\'echet spaces. The map $\mathsf{data}$ is bijective 
because the Cauchy problem for the inhomogeneous Klein-Gordon equation is well-posed, see e.g.\
\cite[Theorem 3.2.12]{BGP} and \cite[Corollary 3.5]{Ginoux}. 
It then follows from the open mapping theorem (see e.g.\ \cite[Corollary 2.12]{Rudin}) 
that the inverse $\mathsf{solve} := \mathsf{data}^{-1}$ is continuous. 
We can define $r := \mathsf{solve}(-,\Phi_0,\Pi_0): C^\infty(M)\to C^\infty(M)$
to be the continuous linear map that assigns solutions $\Phi= \mathsf{solve}(\Phi^\ddagger,\Phi_0,\Pi_0)$ 
to the inhomogeneous Klein-Gordon equation $(\square +m^2)\Phi = \Phi^\ddagger$ for a 
fixed choice of initial data $j^\ast (\Phi) = \Phi_0$ and $j^\ast(\ast\dd_{\dR} \Phi) = \Pi_0$.
\sk

An alternative construction is to use the retarded/advanced Green's operators
$G^\pm$ for the Klein-Gordon operator $\square +m^2$, which by \cite[Corollary 3.11]{Bar} 
extend continuously to sections with past/future compact support. 
Choosing two Cauchy surfaces $\Sigma,\Sigma^\prime\subset M$
such that $\Sigma^\prime\subset I^{+}_M(\Sigma)$ lies in the chronological future of $\Sigma$,
we obtain an open cover $\{I^{+}_M(\Sigma), I^-_M(\Sigma^\prime)\}$ of $M$. Picking any partition of unity
$\chi_+ + \chi_- =1$ subordinate to this cover, we observe that 
$\chi_\pm$ has past/future compact support by construction. 
Therefore, we can define the continuous linear map
\begin{flalign}
r \,:\, C^\infty(M)\, \longrightarrow~\, C^\infty(M)~,~~\Phi^\ddagger \,\longmapsto \,
G^+(\chi_+\,\Phi^\ddagger)  + G^-(\chi_-\,\Phi^\ddagger)\quad. 
\end{flalign}
Note that this defines
a splitting $(\square +m^2)r(\Phi^\ddagger) = (\square +m^2)\big(G^+(\chi_+\,\Phi^\ddagger) + 
G^-(\chi_-\,\Phi^\ddagger)\big) = (\chi_+ +\chi_-)\Phi^\ddagger = \Phi^\ddagger$ as a consequence of
the property $(\square +m^2)\, G^\pm=\id$ (on sections with past/future compact support) 
of retarded/advanced Green's operators.
This second construction yields continuous algebra M{\o}ller maps $K$ 
which are of the same form as those used in the pAQFT literature, see e.g.\ \cite{Duetsch,Hawkins,RejznerProceedings}.
\end{ex}

In the context of Theorem \ref{theo:FTexistence}, a continuous 
splitting $r$ is equivalent to the datum of a continuous deformation retract 
\begin{equation}
\begin{tikzcd}
(\ker\dd,\,0) \ar[r,shift right=-1ex,"k"] & \ar[l,shift right=-1ex,"q"] (\FFF,\,\dd) \ar[loop,out=-25,in=25,distance=20,swap,"-r"]
\end{tikzcd}
\end{equation}
between the complex $\FFF$ concentrated in degrees $\{0,1\}$ 
and the complex $\ker \dd$ concentrated in degree $0$.
Here $k: \ker \dd \to \FFF$ denotes the continuous cochain map
which is defined by the subspace inclusion $\ker \dd \subseteq \FFF^0$
and the continuous cochain map $q: \FFF \to \ker \dd$ is defined uniquely by 
$k\, q := \id + \partial (-r)$. This yields a continuous left inverse 
of $k: \ker \dd \to \FFF$, i.e.\ $q\, k = \id$, and implies that
$-r$ is a continuous homotopy between $k\, q$ and $\id$. Note that this equivalence 
between continuous splittings $r$ and continuous
deformation retracts $(k,q,-r)$ relies crucially on the assumption 
that $\FFF$ is concentrated in degrees $\{0,1\}$. 
Because the complexes considered in the rest of the paper will 
not belong to this class (in particular, they will have non-trivial ghosts in degree $\FFF^{-1}$), 
from now on the relevant concept shall be that of a continuous deformation retract. 
(See also the related comment at the beginning of this section.)
The following theorem is an analytical refinement of Theorem \ref{theo:nonexistence}
in the finite-dimensional setting.
\begin{theo}\label{theo:FTnonexistence}
Suppose that $(\FFF,\ell)$ is a field theory $L_\infty[1]$-algebra whose underlying
cochain complex
\begin{flalign}
\FFF\,=\,\Big(
\xymatrix{
\FFF^{-1}\ar[r]^-{\dd}\,&\, \FFF^0 \ar[r]^-{\dd}\,&\, \cdots  \ar[r]^-{\dd}\,&\, \FFF^N
}
\Big)
\end{flalign}
is concentrated in degrees $\{-1,0,\dots,N\}$ for some $-1\leq N\in\bbZ$.
Assume that there exists a continuous deformation retract
\begin{equation}
\begin{tikzcd}
(\mathsf{H}\FFF,\,0) \ar[r,shift right=-1ex,"k"] & \ar[l,shift right=-1ex,"q"] (\FFF,\,\dd) \ar[loop,out=-25,in=25,distance=20,swap,"w"]
\end{tikzcd}
\end{equation}
to a cochain complex $\mathsf{H}\FFF = (\mathsf{H}\FFF,\,0)$ of complete locally convex topological 
vector spaces with trivial differential. If the transfer 
\begin{flalign}\label{eqn:transferell2}
\widetilde{\ell}_2 \,:\,\xymatrix@C=2em{
\textstyle{\widehat{\bigwedge}}^2\,\mathsf{H}\FFF^{-1} \ar[r]^-{k \wedge k} \,&\, \textstyle{\widehat{\bigwedge}}^2 \FFF^{-1} \,=\, \big(\widehat{\Sym}^2 \FFF\big)^{-2} \ar[r]^-{\ell_2}\,&\, \FFF^{-1} \ar[r]^-{q}\,&\,\mathsf{H}\FFF^{-1}
}\quad,
\end{flalign}
of the $\ell_2$ bracket
to $\mathsf{H}\FFF^{-1}$ is non-vanishing, then there does not exist a M{\o}ller map.
\end{theo}
\begin{proof}
The proof strategy mimics that of Theorem \ref{theo:nonexistence}. 
Transposition of the continuous deformation retract $(k,q,w)$ 
yields a continuous deformation retract 
$(i,p,h) := (q^\prime,k^\prime,w^\prime)$ 
of the strong dual complex $\FFF^\prime$ 
onto the strong dual complex $\mathsf{H}\FFF^\prime$.
Note that the latter has a trivial differential
because the differential of $\mathsf{H}\FFF$ is by hypothesis trivial. 
\sk

By the same construction as in \cite[Proposition 2.5.5]{Gwilliam} 
or \cite[Section 4.1]{Bruinsma}, the continuous deformation retract 
$(i,p,h)$ of $\FFF^\prime$ onto $\mathsf{H}\FFF^\prime$ 
can be lifted along $\widehat{\Sym}$ 
(see also Remark \ref{rem:completedSym}) 
and extended to formal power series, 
which yields a continuous deformation retract
\begin{equation}\label{eqn:freeSDR2}
\begin{tikzcd}
\big(\widehat{\Sym}\,\mathsf{H}\FFF^\prime[[\lambda]],\, 0\big) \ar[r,shift right=-1ex,"I"] & \ar[l,shift right=-1ex,"P"]
\big(\widehat{\Sym}\,\FFF^\prime[[\lambda]],\,\delta_0\big) \ar[loop,out=-12,in=12,distance=20,swap,"H"]
\end{tikzcd}\quad.
\end{equation}
Applying homological perturbation theory \cite{HPT} to the formal 
perturbation $\delta = \sum_{n= 2}^\infty\lambda^{n-1}\,\ell_n^\prime$ 
of the differential $\delta_0$ of $\widehat{\Sym}\, \FFF^\prime[[\lambda]]$, 
see \eqref{eqn:CEdifferential2}, 
yields a perturbed continuous deformation retract
\begin{equation}\label{eqn:interactingSDR2}
\begin{tikzcd}
\big(\widehat{\Sym}\,\mathsf{H}\FFF^\prime[[\lambda]],\, \widetilde{\delta}\,\big) \ar[r,shift right=-1ex,"\widetilde{I}"] 
& \ar[l,shift right=-1ex,"\widetilde{P}"]\big( \widehat{\Sym}\,\FFF^\prime[[\lambda]], \,\delta_0 + \delta\big) \ar[loop,out=-10,in=10,distance=21,swap,"\widetilde{H}"]
\end{tikzcd}\quad.
\end{equation}
Explicit expressions for $\widetilde{\delta}, \widetilde{I}, \widetilde{P}, \widetilde{H}$ 
in terms of $\delta, I, P, H$ can be found in \cite{HPT}.
\sk

As in the proof of Theorem \ref{theo:nonexistence}, 
it follows from the two continuous deformation retracts 
\eqref{eqn:freeSDR2} and \eqref{eqn:interactingSDR2} 
that the existence of a M{\o}ller map
$K: \big(\widehat{\Sym}\,\FFF^\prime[[\lambda]],\,\delta_0\big)\to 
\big(\widehat{\Sym}\,\FFF^\prime[[\lambda]],\,\delta_0+\delta\big)$ 
would entail that the perturbed differential $\widetilde{\delta} = 0$ 
vanishes. Hence, one can prove that there does {\em not} exist 
a M{\o}ller map by showing that the perturbed differential 
$\widetilde{\delta} \neq 0$ does not vanish.
\sk

Recall from the proof of Theorem \ref{theo:nonexistence} that 
the leading term of the perturbed differential $\widetilde{\delta}$ 
is given by $\widetilde{\delta}_1 := P\, \ell_2^\prime\, I$. 
Restricting the leading term $\widetilde{\delta}_1$ to 
$(\mathsf{H}\FFF^\prime)^1 \subseteq (\widehat{\Sym}\, \mathsf{H}\FFF^\prime)^1$ 
yields the continuous linear map
\begin{flalign}
\widetilde{\delta}_1\vert_{(\mathsf{H}\FFF^\prime)^1} \,:\, \xymatrix@C=2em{
(\mathsf{H}\FFF^\prime)^1 \ar[r]^-{i} \,&\,(\FFF^\prime)^1 \ar[r]^-{\ell_2^\prime}\,&\, 
\big(\widehat{\Sym}^2\,\FFF^{\prime}\big)^2\,=\,\textstyle{\widehat{\bigwedge}^2}\,(\FFF^{\prime})^1 
\ar[r]^-{p\wedge p}\,&\,\textstyle{\widehat{\bigwedge}^2}\,\big(\mathsf{H}\FFF^\prime\big)^1
}\quad,
\end{flalign}
where the equality follows from the fact that $\FFF^\prime$ 
is concentrated in degrees $\{-N,\ldots,0,1\}$ because 
$\FFF$ is by hypothesis concentrated in degrees $\{-1,0,\ldots,N\}$.
\sk

It remains to identify $\widetilde{\delta}_1\vert_{(\mathsf{H}\FFF^\prime)^1}$ 
with the transpose of the transfer $\widetilde{\ell}_2$ of the $\ell_2$ 
bracket to $\mathsf{H}\FFF^{-1}$, see \eqref{eqn:transferell2}. 
For this we recall that $\FFF^{-1}$ is always a nuclear Fr\'echet space in our setup,
and hence so is the closed subspace $\ker(\dd: \FFF^{-1} \to \FFF^0)\subseteq \FFF^{-1}$.
(Metrizability and completeness 
are manifestly preserved when restricting to a closed subspace. 
For nuclearity see e.g.\ \cite[Corollary 21.2.3]{Jarchow}.)
Recalling that $\FFF$ is concentrated in degrees $\{-1,0,\ldots,N\}$ 
and that $\mathsf{H}\FFF$ has a trivial differential, 
one realizes that the continuous deformation retract $(k,q,w)$ yields 
a topological linear isomorphism $\mathsf{H}\FFF^{-1} \cong \ker(\dd: \FFF^{-1} \to \FFF^0)$ 
given by the corestriction of $k$. It then follows that $\mathsf{H}\FFF^{-1}$ 
is a nuclear Fr\'echet space too. Using a result of Grothendieck, 
see e.g.\ \cite[Theorem 21.5.9]{Jarchow}, we obtain topological linear isomorphisms
$\textstyle{\widehat{\bigwedge}^2}\,\big(\FFF^\prime\big)^1 = \textstyle{\widehat{\bigwedge}^2}\,\big(\FFF^{-1}\big)^\prime \cong \big(\textstyle{\widehat{\bigwedge}^2}\,\FFF^{-1}\big)^\prime$ and 
$\textstyle{\widehat{\bigwedge}^2}\,(\mathsf{H}\FFF^\prime)^1 = \textstyle{\widehat{\bigwedge}^2}\,(\mathsf{H}\FFF^{-1})^\prime \cong \big(\textstyle{\widehat{\bigwedge}^2}\,\mathsf{H}\FFF^{-1}\big)^\prime$,
which allow us to identify 
the completed exterior power of the strong dual with the strong dual 
of the completed exterior power. (See also the topological linear isomorphism in 
\eqref{eqn:elltranspose} and the subsequent paragraph explaining 
its construction.) 
Through these topological linear isomorphisms we recognize that the continuous linear map 
$\widetilde{\delta}_1\vert_{(\mathsf{H}\FFF^\prime)^1} = {\widetilde{\ell}_2}^\prime$ 
is the transpose of $\widetilde{\ell}_2$. 
Because $\widetilde{\ell}_2 \neq 0$ is non-vanishing by hypothesis, 
$\widetilde{\delta}_1\vert_{(\mathsf{H}\FFF^\prime)^1} \neq 0$ and 
hence also $\widetilde{\delta} \neq 0$ does not vanish. 
Therefore, M{\o}ller maps do not exist under the hypotheses of this theorem.
\end{proof}

\begin{ex}\label{ex:CSnonexistence} 
The goal of this example is to show that Theorem \ref{theo:FTnonexistence} 
applies to non-Abelian Chern-Simons theory from Examples \ref{ex:CS} and \ref{ex:CS2}
on a product manifold $M = \bbR \times \Sigma$, where $\Sigma$ is a closed
(i.e.\ compact and without boundary) $2$-manifold. This will imply that M{\o}ller maps do not exist 
for non-Abelian Chern-Simons theory, treated perturbatively around 
the trivial background solution, on such product $3$-manifolds $M=\bbR \times \Sigma$.
(See Remark \ref{rem:noncompact} below for a comment on the case of non-compact $\Sigma$.)
\sk

Our strategy is to construct the continuous deformation retract required by 
Theorem \ref{theo:FTnonexistence} as the composition of two simpler continuous deformation
retracts. In the first step, we use the Poincar{\'e} lemma 
to construct a continuous deformation retract from $\FFF_{\mathrm{CS}}$ 
to the cochain complex of nuclear Fr\'echet spaces 
\begin{flalign}
\FFF^\Sigma_{\mathrm{CS}}\,:=\,\Big(\xymatrix{
\stackrel{(-1)}{\Omega^0(\Sigma,\g)} \ar[r]^-{\dd_{\dR}}\,&\,\stackrel{(0)}{\Omega^1(\Sigma,\g)}
\ar[r]^-{\dd_{\dR}} \,&\, \stackrel{(1)}{\Omega^2(\Sigma,\g)}
}\Big)
\end{flalign}
which describes the phase space of non-Abelian Chern-Simons theory on the $2$-manifold $\Sigma$.
In the second step, we use Hodge theory on $\Sigma$
(which is where closedness of $\Sigma$ becomes important)
in order to construct a continuous deformation retract from $\FFF^\Sigma_{\mathrm{CS}}$ to 
the shifted de Rham cohomology $\mathrm{H}_{\dR}^{\bullet+1}(\Sigma)\otimes\g$ of $\Sigma$.
\sk

To build the first continuous deformation retract, let us denote by
$j: \Sigma \hookrightarrow M = \bbR\times \Sigma$ the smooth embedding 
associated with the inclusion $\{0\} \times \Sigma \subset M$, by
$\pi: M  = \bbR\times\Sigma \to \Sigma$ the smooth projection onto the second factor, 
and  by $k: \bbR \times M \to M$ the smooth map that sends $(s,(t,x))$ to $((1-s)\,t,x)$. 
One immediately checks that $(j,\pi,k)$ defines a smooth deformation 
retract of $M$ onto $\Sigma$, i.e.\ $\pi\, j = \id$, $k(0,-) = \id$ 
and $k(1,-) = j\, \pi$. Pull-back of differential forms along $\pi$ 
and respectively $j$ yields the continuous cochain maps 
$i := \pi^\ast: \FFF_{\mathrm{CS}}^\Sigma \to \FFF_{\mathrm{CS}}$ 
and respectively
$p:= j^\ast: \FFF_{\mathrm{CS}} \to \FFF_{\mathrm{CS}}^\Sigma$,
which satisfy $p\, i = \id$. 
Furthermore, pull-back of differential forms along $k$ and fiber integration along
the projection $[0,1] \times M \subseteq \bbR \times M \to M$ 
onto the second factor yield the continuous cochain homotopy 
\begin{flalign}
h\,:\, \xymatrix{
\FFF_{\mathrm{CS}}^\bullet = \Omega^{\bullet+1}(M,\g) \ar[r]^-{k^\ast} \,&\, \Omega^{\bullet+1}(\bbR \times M,\g) \ar[r]^-{\int_0^1} \,&\, \Omega^\bullet(M,\g) = \FFF_{\mathrm{CS}}^{\bullet-1}
}
\end{flalign}
comparing $i\, p$ and $\id$, namely $\partial h = i\, p - \id$. 
Summing up, the above defines a continuous deformation 
retract
\begin{equation}\label{eqn:CSdefret1}
\begin{tikzcd}
(\FFF_{\mathrm{CS}}^\Sigma,\,\dd)\ar[r,shift right=-1ex,"i"] 
& \ar[l,shift right=-1ex,"p"] (\FFF_{\mathrm{CS}},\,\dd) \ar[loop,out=-15,in=15,distance=20,swap,"h"]
\end{tikzcd}
\end{equation}
of $\FFF_{\mathrm{CS}}$ onto $\FFF_{\mathrm{CS}}^\Sigma$. 
\sk

To build the second continuous deformation retract, we choose 
a Riemannian metric on the closed $2$-manifold $\Sigma$.
From this we can define the inner products
$\langle\alpha,\beta\rangle := \int_{\Sigma}\alpha\wedge\ast\beta$,
the adjoint differentials $\langle \dd_\dR^\ast\alpha,\beta \rangle := \langle \alpha,\dd_{\dR}\beta\rangle$
and the Laplace operators $\Delta := \dd_{\dR}^\ast\,\dd_{\dR} + \dd_{\dR}\,\dd_{\dR}^\ast$
on the Fr\'echet spaces of $k$-forms $\Omega^k(\Sigma)$, for all $k=0,1,2$.
\sk

For each $k$, the Hodge decomposition theorem (see e.g.\ \cite{Warner1983,Wells2007})
applies to give $\Omega^k(\Sigma)$ as an orthogonal direct sum
\begin{flalign} \label{eqn:Hodgedecomp}
\Omega^k(\Sigma)
\,=\, \mathcal{H}^k(\Sigma) \oplus \mathrm{im}\,\Delta
\,=\, \mathcal{H}^k(\Sigma) \oplus \mathrm{im}\, \dd_{\dR} \oplus \mathrm{im}\, \dd_{\dR}^\ast
\quad,
\end{flalign}
where $\mathcal{H}^k(\Sigma) := \ker \Delta$ is the space of harmonic $k$-forms.
Moreover, it follows that for each $k$ a Green's operator
$G : \Omega^k(\Sigma)\to \Omega^k(\Sigma)$ exists satisfying
$G\,\Delta = \Delta\,G = \id - i_{\mathcal{H}}\,p_{\mathcal{H}}$,
where $p_{\mathcal{H}}$ denotes the projection from
$\Omega^k(\Sigma)$ onto $\mathcal{H}^k(\Sigma)$ and $i_{\mathcal{H}}$ denotes
the inclusion of $\mathcal{H}^k(\Sigma)$  into $\Omega^k(\Sigma)$.
These Green's operators commute with the differentials, i.e.\
$G \,\dd_{\dR} = \dd_{\dR}\, G$ and $G\, \dd_{\dR}^\ast = \dd_{\dR}^\ast\, G$.
Combining these data we obtain an algebraic deformation retract 
\begin{equation}\label{eqn:CSdefret2}
\begin{tikzcd}
\big(\mathsf{H}_{\dR}^{\bullet+1}(\Sigma)\otimes\g,\,0\big)\,\cong\,\big(\mathcal{H}^{\bullet+1}(\Sigma)\otimes\g,\,0\big) \ar[r,shift right=-1ex,"i_{\mathcal{H}}"] 
& \ar[l,shift right=-1ex,"p_{\mathcal{H}}"]\big( \FFF^\Sigma_{\mathrm{CS}}, \dd\big) \ar[loop,out=-15,in=15,distance=23,swap,"-G\,\dd_\dR^{\ast}"]
\end{tikzcd}
\end{equation}
and it remains to prove that this deformation retract is continuous with respect to the Fr\'echet topologies.
\sk

We note first that the Hodge decomposition \eqref{eqn:Hodgedecomp} is a topological 
direct sum decomposition of Fr\'echet spaces,
i.e.\ each summand is Fr\'echet and the projection onto and inclusion of each summand is a continuous linear map.
This holds because the Hodge decomposition is orthogonal, so each summand is the 
orthogonal complement of the (sum of the) others.
The orthogonal complement $A^\perp = \bigcap_{\alpha\in A} \ker\langle\alpha,-\rangle$ of 
any $A \subseteq \Omega^k(\Sigma)$ is the intersection of kernels of continuous linear maps,
which implies that it is closed in $\Omega^k(\Sigma)$ with respect to the Fr\'echet topology and therefore Fr\'echet itself.
It then follows from \cite[Section 15.12 (6)]{Koethe} that the Hodge decomposition \eqref{eqn:Hodgedecomp} is a topological
direct sum decomposition.
\sk

It remains to show that the Green's operators are continuous.
This follows
from the open mapping theorem (see e.g.\ \cite[Corollary 2.12]{Rudin}) and the fact
that $G = i_{\mathcal{H}^\perp}\,\Delta_\perp^{-1}\,p_{\mathcal{H}^\perp}$,
where $\Delta_\perp^{-1}$ denotes the inverse of the bijective continuous linear map $\Delta_\perp : 
\mathcal{H}^k(\Sigma)^\perp\to \mathrm{im}\,\Delta = \mathcal{H}^k(\Sigma)^\perp$ 
that is obtained by restricting and co-restricting $\Delta$ to the orthogonal complement of harmonic forms,
which is Fr\'echet as noted above. Here $i_{\mathcal{H}^\perp}$ and $p_{\mathcal{H}^\perp}$ denote 
the continuous projection to and inclusion of $\mathcal{H}^k(\Sigma)^\perp$ in $\Omega^k(\Sigma)$.
\sk

Composing the two continuous deformation retracts \eqref{eqn:CSdefret1} and \eqref{eqn:CSdefret2},
we obtain a continuous deformation retract from $\FFF_{\mathrm{CS}}$ onto $\mathcal{H}^{\bullet+1}(\Sigma)\otimes \g$.
The transfer of the Chern-Simons $\ell_2$ bracket from Example \ref{ex:CS2} to
$\mathcal{H}^{\bullet+1}(\Sigma)\otimes \g$ is then given by
\begin{flalign} \label{eqn:CS:transferred_l2}
\ell^\mathcal{H}_2(\alpha,\beta)\,:=\, (-1)^{\vert\alpha\vert}\,p_{\mathcal{H}}\big(j^\ast[\pi^\ast\alpha , \pi^\ast\beta]\big)
\,=\,(-1)^{\vert\alpha\vert}\,p_{\mathcal{H}}\big(\big[j^\ast\pi^\ast\alpha,j^\ast\pi^\ast\beta\big]\big) \,=\, (-1)^{\vert\alpha\vert}\,p_\mathcal{H} [\alpha,\beta]\quad,
\end{flalign}
for all $\alpha,\beta\in \mathcal{H}^{\bullet+1}(\Sigma)\otimes \g$.
In the second step
we used that the Lie bracket on $\g$-valued forms commutes with pull-backs of forms and
in the third step we used that $j^\ast\pi^\ast = \id$.
Evaluating $\ell_2^\mathcal{H}$ on degree $-1$ elements 
$\alpha , \beta \in \mathcal{H}^{-1+1}(\Sigma)\otimes \g = \mathcal{H}^{0}(\Sigma)\otimes \g$
gives the transferred bracket $\widetilde{\ell}_2$ of \eqref{eqn:transferell2},
which simplifies further to $\widetilde{\ell}_2(\alpha, \beta) = - [\alpha, \beta]$.
This is because $\alpha$, $\beta$ and hence also $[\alpha, \beta]$ are locally 
constant $\g$-valued functions on $\Sigma$, so the harmonic projector $p_\mathcal{H}$ acts as the identity.
This shows that all hypotheses of Theorem \ref{theo:FTnonexistence}
are satisfied by non-Abelian Chern-Simons theory on $M=\bbR\times \Sigma$ with $\Sigma$ closed, 
hence there does not exist a M{\o}ller map for this theory.
\end{ex}

\begin{ex}\label{ex:YMnonexistence}
The goal of this example is to show that Theorem \ref{theo:FTnonexistence} 
applies to non-Abelian Yang-Mills theory from Examples \ref{ex:YM} and \ref{ex:YM2}
on a globally hyperbolic Lorentzian manifold $M$ with a compact spacelike Cauchy surface
$j: \Sigma \hookrightarrow M$. This will imply that M{\o}ller maps do not exist 
for non-Abelian Yang-Mills theory, treated perturbatively around the trivial 
background solution, on such globally hyperbolic Lorentzian manifolds $M$.
(See Remark \ref{rem:noncompact} below for a comment on the case of non-compact $\Sigma$.)
Similarly to Example \ref{ex:CSnonexistence}, our strategy is to construct
the continuous deformation retract required by Theorem \ref{theo:FTnonexistence}
as the composition of two simpler continuous deformation retracts. Since
non-Abelian Yang-Mills theory is richer than the topological non-Abelian 
Chern-Simons theory, the specific form of these two continuous deformation retracts
will be more involved than in the case of Example \ref{ex:CSnonexistence}.
\sk

To build the first continuous deformation retract, we 
use well-posedness of the initial value problem for linearized Yang-Mills theory.
For this we consider the cochain complex of nuclear Fr\'echet spaces 
\begin{flalign}\label{eqn:YMinitialdata}
\FFF^\Sigma_{\mathrm{YM}}\,:=\,\Big(\xymatrix@C=3em{
\stackrel{(-1)}{\Omega^0(\Sigma,\g)} \ar[r]^-{(\dd_{\dR},0)}\,&\,\stackrel{(0)}{\Omega^1(\Sigma,\g) \oplus \Omega^{m-2}(\Sigma,\g)}
\ar[r]^-{\dd_{\dR}\, \pr_2} \,&\, \stackrel{(1)}{\Omega^{m-1}(\Sigma,\g)}
}\Big)
\end{flalign}
which describes the complex of initial data on the Cauchy surface $\Sigma$.
We define the continuous cochain map 
\begin{subequations}
\begin{flalign}
p\,:\, \FFF_{\mathrm{YM}} ~\longrightarrow~ \FFF^\Sigma_{\mathrm{YM}} 
\end{flalign}
by the non-trivial components
\begin{flalign}
p^{-1} \,:=\, j^\ast\quad,\quad p^0 \,:=\, \big(j^\ast,\,j^\ast \ast \dd_{\dR}\big)\quad,\quad 
p^1 \,:=\, j^\ast\quad.
\end{flalign}
\end{subequations}
To construct the continuous cochain map 
$i: \FFF^\Sigma_{\mathrm{YM}} \to \FFF_{\mathrm{YM}}$, 
we consider the continuous initial data map between Fr\'echet spaces
\begin{flalign}
\nn \mathsf{data}_k\,:\,\Omega^k(M) \,&\longrightarrow\, \Omega^k(M) \times \Omega^k(\Sigma) \times \Omega^{m-k}(\Sigma) \times \Omega^{m-k-1}(\Sigma) \times \Omega^{k-1}(\Sigma)\quad, \\[4pt]
\omega \,&\longmapsto\, \big( \square \omega,\, j^\ast (\omega),\, j^\ast(\ast\, \omega),\, j^\ast(\ast\, \dd_{\dR} \omega),\, j^\ast(\ast\, \dd_{\dR} \ast \omega) \big) \quad,
\end{flalign}
which is associated with the inhomogeneous Cauchy problem for the wave operator
$\square$ on $k$-forms on $M$, for $k = 0, \ldots, m$. 
(Our convention for the wave operator is the standard one 
$\square := \dd_\dR^\ast \,\dd_\dR + \dd_\dR\,\dd_\dR^\ast$,
with the adjoint differential defined on $k$-forms by $\dd_\dR^\ast := (-1)^{m(k+1)}\ast\,\dd_\dR\,\ast$.)
This map is bijective because the Cauchy problem is well-posed, 
see e.g.\ \cite[Theorem 3.2.12]{BGP} and \cite[Corollary 3.5]{Ginoux}. 
It then follows from the open mapping theorem 
(see e.g.\ \cite[Corollary 2.12]{Rudin}) that the inverse 
$\mathsf{solve}_k := \mathsf{data}_k^{-1}$ is a continuous linear map. 
With these preparations, we define the continuous cochain map 
\begin{subequations}
\begin{flalign}
i\,:\, \FFF^\Sigma_{\mathrm{YM}} ~\longrightarrow~ \FFF_{\mathrm{YM}} 
\end{flalign}
by the non-trivial components
\begin{flalign}
i^{-1} \,&:=\, \mathsf{solve}_0 \circ (0,\id,0,0,0)\quad,\\
i^0 \,&:=\, \mathsf{solve}_1 \circ (0,\pr_1,0,\pr_2,0)\quad,\\
i^1 \,&:=\, \mathsf{solve}_{m-1} \circ (0,\id,0,0,0)\quad.
\end{flalign}
\end{subequations} 
One directly checks that $p\, i = \id$. To complete the construction 
of a continuous deformation retract, let us define the continuous cochain homotopy 
\begin{subequations}
\begin{flalign}
h\,:\, \FFF_{\mathrm{YM}}^\bullet ~\longrightarrow~ \FFF_{\mathrm{YM}}^{\bullet-1}
\end{flalign}
by the non-trivial components
\begin{flalign}
h^0 \,&:=\, - \,\mathsf{solve}_0 \circ (\ast \dd_{\dR} \ast,\,0,\, 0,\, j^\ast\, \ast,\,0)\quad,\\ 
h^1 \,&:=\, (-1)^{m-1}~ \mathsf{solve}_1 \circ (\ast,\, 0,\, 0,\, 0,\, 0)\quad,\\
h^2 \,&:=\, -\,\mathsf{solve}_{m-1} \circ (\ast \dd_{\dR} \ast,\, 0,\, 0,\, j^\ast\, \ast,\,0)\quad.
\end{flalign}
\end{subequations} 
By direct inspection one shows that $\partial h = i\, p - \id$. 
(As a side-remark, let us observe that the first entries of the 
maps which pre-compose the $\mathsf{solve}$ maps for $\square$ 
are the Green witnesses from \cite{Benini}.)
Summing up, the above defines a continuous deformation 
retract
\begin{equation}\label{eqn:YMdefret1}
\begin{tikzcd}
(\FFF_{\mathrm{YM}}^\Sigma,\,\dd)\ar[r,shift right=-1ex,"i"] 
& \ar[l,shift right=-1ex,"p"] (\FFF_{\mathrm{YM}},\,\dd) \ar[loop,out=-15,in=15,distance=20,swap,"h"]
\end{tikzcd}
\end{equation}
of $\FFF_{\mathrm{YM}}$ onto $\FFF^\Sigma_{\mathrm{YM}}$. 
\sk

To build the second continuous deformation retract, we proceed in analogy to Example \ref{ex:CSnonexistence}
and use the Hodge decomposition
$\Omega^k(\Sigma) = \mathcal{H}^k(\Sigma) \oplus \mathrm{im}\, \dd_\dR \oplus \mathrm{im}\, \dd_\dR^\ast$
on the compact Cauchy surface $\Sigma$.
As noted below \eqref{eqn:CSdefret2}, this is a topological direct sum decomposition of Fr\'echet spaces.
Using this decomposition we define the closed subcomplex
\begin{subequations}
\begin{flalign}
\mathsf{H}\FFF_{\mathrm{YM}}^\Sigma \,\subseteq\, \FFF_{\mathrm{YM}}^\Sigma
\end{flalign}
of the cochain complex of nuclear Fr\'echet spaces \eqref{eqn:YMinitialdata} by the components
\begin{flalign}
\mathsf{H}^{-1}\FFF_{\mathrm{YM}}^\Sigma\,&:=\, \mathcal{H}^0(\Sigma)\otimes\g\,\subseteq\, \Omega^0(\Sigma,\g)\quad,\\[3pt]
\nn \mathsf{H}^{0}\FFF_{\mathrm{YM}}^\Sigma\,&:=\, \big((\mathcal{H}^1(\Sigma) \oplus\mathrm{im}\,\dd_\dR^\ast )\oplus (\mathcal{H}^{m-2}(\Sigma) \oplus\mathrm{im}\,\dd_\dR)\big)\otimes\g \\
\,&\,\subseteq\, \Omega^1(\Sigma,\g)\oplus\Omega^{m-2}(\Sigma,\g)\quad,\\[3pt]
\mathsf{H}^{1}\FFF_{\mathrm{YM}}^\Sigma\,&:=\,\mathcal{H}^{m-1}(\Sigma)\otimes\g \,\subseteq\,\Omega^{m-1}(\Sigma,\g)\quad.
\end{flalign}
\end{subequations}
One checks that the differential of $ \FFF_{\mathrm{YM}}^\Sigma$ restricts to the trivial differential on 
$\mathsf{H}\FFF_{\mathrm{YM}}^\Sigma$. Furthermore, the inclusion
$i_{\mathsf{H}} : \mathsf{H}\FFF_{\mathrm{YM}}^\Sigma\to \FFF_{\mathrm{YM}}^\Sigma$
and the projection $p_{\mathsf{H}}: \FFF_{\mathrm{YM}}^\Sigma\to \mathsf{H}\FFF_{\mathrm{YM}}^\Sigma$
are both continuous cochain maps since their components are given by the continuous
inclusions and projections associated with the Hodge decomposition.
Explicitly, the projection $p_{\mathsf{H}}$
can be computed in terms of the harmonic projectors 
$p_{\mathcal{H}}: \Omega^k(\Sigma) \to \mathcal{H}^k(\Sigma)$ 
and the Green's operators $G: \Omega^k(\Sigma) \to \Omega^k(\Sigma)$ 
for the Laplace operators $\Delta = \dd_\dR^\ast\, \dd_\dR + \dd_\dR\, \dd_\dR^\ast: \Omega^k(\Sigma) \to \Omega^k(\Sigma)$ on the Cauchy surface $\Sigma$ by
\begin{subequations}\label{eqn:YM-Hodge-proj}
\begin{flalign}
p^{-1}_{\mathsf{H}} \,&=\, p_{\mathcal{H}}\quad, \\
p^{0}_{\mathsf{H}} \,&=\, (p_{\mathcal{H}} \oplus \dd_\dR^\ast\, G\, \dd_\dR) \oplus (p_{\mathcal{H}} \oplus \dd_\dR\, G\, \dd_\dR^\ast)\quad, \\
p^{1}_{\mathsf{H}} \,&=\, p_{\mathcal{H}}\quad.
\end{flalign}
\end{subequations}
Recalling that the Green's operators are continuous, as noted below \eqref{eqn:CSdefret2},
we obtain a continuous deformation retract
\begin{equation}\label{eqn:YMdefret2}
\begin{tikzcd}
(\mathsf{H}\FFF_{\mathrm{YM}}^\Sigma,\,0)\ar[r,shift right=-1ex,"i_{\mathsf{H}}"] 
& \ar[l,shift right=-1ex,"p_{\mathsf{H}}"] (\FFF_{\mathrm{YM}}^\Sigma,\,\dd) \ar[loop,out=-15,in=15,distance=20,swap,"h_{\mathsf{H}}"]
\end{tikzcd}
\end{equation}
by using the continuous cochain homotopy 
$h_{\mathsf{H}}: \FFF_{\mathrm{YM}}^{\Sigma,\,\bullet}\to \FFF_{\mathrm{YM}}^{\Sigma,\,\bullet-1}$ 
which is defined by the non-trivial components
\begin{subequations}
\begin{flalign}
h_{\mathsf{H}}^0\,&:=\, - \dd_{\dR}^\ast\, G\, \pr_1\quad,\\
h_{\mathsf{H}}^1\,&:=\, - \iota_2\, G\, \dd_{\dR}^\ast\quad.
\end{flalign}
\end{subequations}
The identity  $\partial h_{\mathsf{H}} = i_{\mathsf{H}}\, p_{\mathsf{H}} - \id$
can be verified with a direct computation using also \eqref{eqn:YM-Hodge-proj}.
\sk

Composing the two continuous deformation retracts \eqref{eqn:YMdefret1} and \eqref{eqn:YMdefret2},
we obtain a continuous deformation retract from $\FFF_{\mathrm{YM}}$ onto $\mathsf{H}\FFF_{\mathrm{YM}}^\Sigma$.
The transfer \eqref{eqn:transferell2} of the Yang-Mills $\ell_2$ bracket from Example \ref{ex:YM2} to
$\mathsf{H}^{-1}\FFF_{\mathrm{YM}}^\Sigma = \mathcal{H}^{0}(\Sigma)\otimes \g$ is then given by
\begin{flalign}
\widetilde{\ell}_2(\alpha,\beta)\,=\, -p_{\mathsf{H}}\big(j^\ast\big[i(\alpha), i(\beta) \big]\big)
\,=\,-p_{\mathsf{H}}\big(\big[j^\ast i(\alpha), j^\ast i(\beta) \big]\big) \,=\,-[\alpha,\beta]\quad,
\end{flalign} 
for all harmonic $0$-forms $\alpha,\beta\in \mathcal{H}^{0}(\Sigma)\otimes \g$.
In the second step
we used that the Lie bracket on $\g$-valued forms commutes with pull-backs of forms.
In the third step we used that $j^\ast\,i = p\,i = \id$ and that 
$[\alpha,\beta]$ is locally constant, hence a harmonic $0$-form, so the harmonic projector $p_{\mathsf{H}}$
acts as the identity. This shows that all hypotheses of Theorem \ref{theo:FTnonexistence}
are satisfied by non-Abelian Yang-Mills theory on a globally hyperbolic Lorentzian
manifold $M$ with a compact spacelike Cauchy surface $\Sigma\subset M$,
hence there does not exist a M{\o}ller map for this theory.
\end{ex}

\begin{rem}\label{rem:noncompact}
It is important to emphasize that compactness of the (Cauchy) surface $\Sigma$ 
in Examples \ref{ex:CSnonexistence} and \ref{ex:YMnonexistence}
is only a technical hypothesis, not a conceptual one. It
allows us to systematically construct the second continuous deformation retracts from data
on $\Sigma$ to their cohomology by using powerful techniques from Hodge theory. (Note that, 
in both examples, the first
continuous deformation retract from data on $M$
to data on $\Sigma$ exists regardless of compactness of $\Sigma$.)
Conceptually, our interpretation in Remark \ref{rem:origin}
of these particular obstructions to the existence of M{\o}ller maps
suggests that the compactness assumptions should be inessential.
In order to substantiate this claim, let us provide some simple 
adaptations of Examples \ref{ex:CSnonexistence} and \ref{ex:YMnonexistence} to non-compact $\Sigma$
where we can construct continuous deformation retracts by ad-hoc methods:
\begin{enumerate}
\item For non-Abelian Chern-Simons theory treated perturbatively around 
the trivial background solution on $M=\bbR^3$, one can iteratively use the 
Poincar\'e lemma, and its associated continuous deformation retracts as in 
\eqref{eqn:CSdefret1}, to construct a composable sequence of continuous deformation retracts
\begin{equation}
\begin{tikzcd}[column sep=small]
\g[1]=\Omega^{\bullet+1}(\mathrm{pt},\g) \ar[r,shift right=-0.5ex] 
&\ar[l,shift right=-0.5ex] \Omega^{\bullet+1}(\bbR,\g)\ar[r,shift right=-0.5ex] \ar[loop,out=70,in=110,distance=22]
& \ar[l,shift right=-0.5ex] \Omega^{\bullet+1}(\bbR^2,\g)\ar[r,shift right=-0.5ex] \ar[loop,out=70,in=110,distance=22]
& \ar[l,shift right=-0.5ex] \Omega^{\bullet+1}(\bbR^3,\g) =\FFF_{\mathrm{CS}}\ar[loop,out=70,in=110,distance=22]
\end{tikzcd}
\end{equation}
from $\bbR^3$ to the point $\mathrm{pt}$.
The transfer \eqref{eqn:transferell2} along the composite continuous deformation retract
of the Chern-Simons $\ell_2$ bracket is then given
by the additive inverse of the Lie bracket of $\g$, i.e.\ $\widetilde{\ell}_2(\alpha,\beta)=-[\alpha,\beta]$, 
for all $\alpha,\beta\in\g$. Hence, the hypotheses of Theorem \ref{theo:FTnonexistence}
are satisfied and consequently there does not exist a M{\o}ller map
for non-Abelian Chern-Simons theory on $\bbR^3$.

\item  For non-Abelian Yang-Mills theory treated perturbatively around 
the trivial background solution on the ${m=2}$-dimensional Minkowski spacetime $M=\bbR^2$,
the initial data complex \eqref{eqn:YMinitialdata} simplifies to
\begin{flalign}
\FFF^\Sigma_{\mathrm{YM}}\,=\,\Big(\xymatrix@C=3em{
\stackrel{(-1)}{\Omega^0(\bbR,\g)} \ar[r]^-{(\dd_{\dR},0)}\,&\,\stackrel{(0)}{\Omega^1(\bbR,\g) \oplus \Omega^{0}(\bbR,\g)}
\ar[r]^-{\dd_{\dR}\, \pr_2} \,&\, \stackrel{(1)}{\Omega^{1}(\bbR,\g)}
}\Big)\quad.
\end{flalign}
Observe that this is the direct sum $\Omega^{\bullet+1}(\bbR,\g)\oplus \Omega^\bullet(\bbR,\g)$ of two (shifted)
$\g$-valued de Rham complexes on the real line $\bbR$.
Hence, we can use the Poincar\'e lemma
and its associated continuous deformation retracts as in \eqref{eqn:CSdefret1} to build a continuous deformation
retract from $\FFF^\Sigma_{\mathrm{YM}}$ to $\mathsf{H}\FFF^\Sigma_{\mathrm{YM}} = \g[1] \oplus\g$.
The transferred Yang-Mills $\ell_2$ bracket \eqref{eqn:transferell2} is then given
by $\widetilde{\ell}_2(\alpha,\beta)=-[\alpha,\beta]$, for all $\alpha,\beta\in\g$ in degree $-1$.
Hence, the hypotheses of Theorem \ref{theo:FTnonexistence}
are satisfied and consequently there does not exist a M{\o}ller map
for non-Abelian Yang-Mills theory on the $2$-dimensional Minkowski spacetime $\bbR^2$.
(In fact,
our argument does not use the geometry on $\bbR^2$, so this result holds for any $2$-dimensional
globally hyperbolic Lorentzian manifold $M$ whose underlying manifold is $\bbR^2$.)
\qedhere
\end{enumerate}
\end{rem}

\begin{rem}\label{rem:CSlastexample}
The non-existence result in Theorem \ref{theo:FTnonexistence} identifies a particular kind of obstruction
to the existence of M{\o}ller maps which is directly linked to the degree $-1$ elements, i.e.\ the ghost fields,
in the field theory $L_\infty[1]$-algebra $(\FFF,\ell)$. The goal of this remark is to illustrate through a 
simple example that the perturbatively interacting observables can also differ considerably from the free ones
in degree $0$.
\sk

For concreteness, let us consider non-Abelian Chern-Simons theory 
treated perturbatively around the trivial background solution
on $M = \bbR^2\times\mathbb{S}^1$. Using the Poincar\'e lemma as in \eqref{eqn:CSdefret1} twice
and then Hodge theory as in \eqref{eqn:CSdefret2} on $\mathbb{S}^1$, we obtain a 
composable sequence of continuous deformation retracts
\begin{equation}\label{eqn:CSlastexample}
\begin{tikzcd}[column sep=small]
\mathcal{H}^{\bullet+1}(\mathbb{S}^1)\otimes\g\ar[r,shift right=-0.5ex] 
& \ar[l,shift right=-0.5ex] \Omega^{\bullet+1}(\mathbb{S}^1,\g)\ar[r,shift right=-0.5ex] \ar[loop,out=70,in=110,distance=22]
& \ar[l,shift right=-0.5ex] \Omega^{\bullet+1}(\bbR\times\mathbb{S}^1,\g)\ar[r,shift right=-0.5ex] \ar[loop,out=70,in=110,distance=22]
& \ar[l,shift right=-0.5ex] \Omega^{\bullet+1}(\bbR^2\times\mathbb{S}^1,\g) =\FFF_{\mathrm{CS}}\ar[loop,out=70,in=110,distance=22]
\end{tikzcd}
\end{equation}
onto the $\g$-valued harmonic forms $\mathcal{H}^{\bullet+1}(\mathbb{S}^1) \otimes\g
= \g[1]\oplus\g$ on $\mathbb{S}^1$. Transposing and extending 
the composite continuous deformation retract to observables, we obtain a 
continuous deformation retract
\begin{equation}\label{eqn:CSlastexample2}
\begin{tikzcd}
\big(\Sym\,(\g^{\ast}[-1]\oplus \g^\ast)[[\lambda]],\, 0\big) \ar[r,shift right=-1ex] & \ar[l,shift right=-1ex]
\big(\widehat{\Sym}\,\FFF_{\mathrm{CS}}^\prime[[\lambda]],\,\delta_0\big) \ar[loop,out=-12,in=12,distance=22,swap]
\end{tikzcd}
\end{equation}
between the free observables of non-Abelian Chern-Simons theory and the 
symmetric algebra of the finite-dimensional cochain complex $\g^{\ast}[-1]\oplus \g^\ast$ with trivial differential. 
In particular, the $0$-th cohomology of the free Chern-Simons observables
is simply given by the symmetric algebra $\Sym\,\g^\ast[[\lambda]]$.
To compare this with the perturbatively interacting Chern-Simons observables,
we apply homological perturbation theory \cite{HPT} to the formal 
perturbation $\delta = \lambda\,\ell_2^\prime$ of \eqref{eqn:CSlastexample2} 
which is associated with the Chern-Simons 
$L_{\infty}[1]$-algebra structure from Example \ref{ex:CS2} and obtain the perturbed
continuous deformation retract
\begin{equation}\label{eqn:CSlastexample3}
\begin{tikzcd}
\big(\Sym\,(\g^{\ast}[-1]\oplus \g^\ast)[[\lambda]],\, \widetilde{\delta}\,\big) \ar[r,shift right=-1ex] & \ar[l,shift right=-1ex]
\big(\widehat{\Sym}\,\FFF_{\mathrm{CS}}^\prime[[\lambda]],\,\delta_0 +\delta\big) \ar[loop,out=-12,in=14,distance=25,swap]
\end{tikzcd}\quad.
\end{equation}
Using the continuous deformation retracts from \eqref{eqn:CSlastexample} and the 
explicit formula for the perturbed differential
$\widetilde{\delta}$ from \cite{HPT}, one finds that $\widetilde{\delta}$ 
is precisely of order $\lambda^1$, i.e.\ there are no higher order
corrections in $\lambda$. (This is due to the fact that all cochain homotopies
in \eqref{eqn:CSlastexample} act trivially on pull-backs of harmonic forms on $\mathbb{S}^1$.)
One further shows that $\widetilde{\delta}$ is the degree $1$ derivation induced by
($\lambda$ times) the transpose of the transfer $\ell^\mathcal{H}_2$ of the 
$\ell_2$ bracket to $\mathcal{H}^{\bullet+1}(\mathbb{S}^1) \otimes \g = \g[1] \oplus \g$,
see also \eqref{eqn:CS:transferred_l2}. The transferred bracket further simplifies
to $\ell^\mathcal{H}_2 (\alpha, \beta) = (-1)^{\vert\alpha\vert}\,p_{\mathcal{H}}[\alpha, \beta]
=(-1)^{\vert\alpha\vert}\, [\alpha, \beta]$, for all 
$\alpha,\beta\in\mathcal{H}^{\bullet+1}(\mathbb{S}^1) \otimes \g$, because
the harmonic forms on $\mathbb{S}^1$ are simply the forms with constant coefficients,
and so are closed under the Lie bracket on $\g$-valued forms.
It then follows that the left-hand side of \eqref{eqn:CSlastexample3} is cochain isomorphic to
the Chevalley-Eilenberg complex
\begin{flalign}\label{eqn:CSlastexample4}
\big(\Sym\,(\g^{\ast}[-1]\oplus \g^\ast)[[\lambda]],\, \widetilde{\delta}\,\big) \,\cong\,
\mathrm{CE}^\bullet\big(\g[[\lambda]],\Sym\,\g^\ast[[\lambda]]\big)
\end{flalign}
of the Lie algebra $\big(\g[[\lambda]],\lambda\,[-,-]\big)$, where the Lie bracket is rescaled
by the formal parameter $\lambda$, with coefficients in the representation given by extending the 
coadjoint action $\lambda\,\mathrm{ad}^\ast$ on $\g^\ast[[\lambda]]$ to $\Sym\,\g^\ast[[\lambda]]$.
Hence, the perturbatively interacting observables for non-Abelian Chern-Simons theory
on $\bbR^2\times\mathbb{S}^1$ are equivalently described by this Chevalley-Eilenberg complex.
\sk

To compare the free observables \eqref{eqn:CSlastexample2} with the perturbatively interacting ones 
\eqref{eqn:CSlastexample3}, we compute the cohomology of the Chevalley-Eilenberg complex
\eqref{eqn:CSlastexample4}.
Recalling that the differential $\widetilde{\delta}$ is precisely of order $\lambda^1$ 
and spelling out the closedness and exactness conditions for $\widetilde{\delta}$ 
for formal power series in $\lambda$, one immediately observes that
\begin{multline}\label{eqn:CSlastexample5}
\mathrm{H}_{\mathrm{CE}}^\bullet\big(\g[[\lambda]],\Sym\,\g^\ast[[\lambda]]\big)\\[5pt]
\,\cong\,\ker\Big(\widetilde{\delta} : \Sym\,(\g^{\ast}[-1]\oplus \g^\ast)\to \lambda\,\Sym\,(\g^{\ast}[-1]\oplus \g^\ast)\Big)
\oplus \lambda\,\mathrm{H}_{\mathrm{CE}}^\bullet\big(\g,\Sym\,\g^\ast\big)[[\lambda]]\quad,
\end{multline}
where $\mathrm{H}_{\mathrm{CE}}^\bullet\big(\g,\Sym\,\g^\ast\big)$ denotes the Chevalley-Eilenberg cohomology
of the unrescaled Lie algebra $\big(\g,\,[-,-]\big)$ with coefficients in the representation given by extending 
the unrescaled coadjoint action $\mathrm{ad}^\ast$ on $\g^\ast$ to $\Sym\,\g^\ast$.
The latter can be computed by decomposing $\Sym\,\g^\ast$ into irreducible representations
and using the Whitehead lemma \cite[Theorem 7.8.9]{Weibel}. We obtain
\begin{flalign}\label{eqn:CSlastexample6}
\mathrm{H}_{\mathrm{CE}}^\bullet\big(\g,\Sym\,\g^\ast\big) \,\cong 
\,\mathrm{H}_{\mathrm{CE}}^\bullet\big(\g,(\Sym\,\g^\ast)^{\g}\big)\,\cong\, 
(\Sym\,\g^\ast)^{\g}\otimes \mathrm{H}_{\mathrm{CE}}^\bullet(\g,\bbK)\quad,
\end{flalign}
where $(\Sym\,\g^\ast)^{\g}\subseteq \Sym\,\g^\ast$ denotes the subspace of $\g$-invariants.
Using further that $\mathrm{H}_{\mathrm{CE}}^0(\g,\bbK)\cong \bbK$, 
we obtain for the $0$-th cohomology \eqref{eqn:CSlastexample5} of the perturbatively interacting Chern-Simons 
observables
\begin{flalign}
\mathrm{H}_{\mathrm{CE}}^0\big(\g[[\lambda]],\Sym\,\g^\ast[[\lambda]]\big)\,\cong\,(\Sym\,\g^\ast)^{\g}[[\lambda]]\quad.
\end{flalign}
Note that this is manifestly different from the $0$-th cohomology $\Sym\,\g^\ast[[\lambda]]$
of the free observables \eqref{eqn:CSlastexample2} since the coadjoint action is non-trivial.
As a concrete example, we observe that for $\g=\mathfrak{su}(2)$, the $0$-th cohomology of the
perturbatively interacting observables $(\Sym\,\mathfrak{su}(2)^\ast)^{\mathfrak{su}(2)}[[\lambda]]
\cong \Sym\,\bbK[[\lambda]]$ is freely generated by $1$ generator, namely the dual 
of the quadratic Casimir, see e.g.\ \cite[Section V.5]{Knapp2002},
while the
free observables  $\Sym\,\mathfrak{su}(2)^\ast[[\lambda]]$ are freely generated by $3$ generators, 
namely a basis for $\mathfrak{su}(2)^\ast$.
\end{rem}


\section*{Acknowledgments}
We would like to thank Federico Bambozzi, Nicola Pinamonti and Jochen Zahn 
for fruitful discussions and comments. 
We also would like to thank the anonymous reviewer for valuable
comments and suggestions which have helped us to improve our paper.
The work of M.B.\ and G.M.\ is supported in part by the MUR Excellence 
Department Project awarded to Dipartimento di Matematica, 
Università di Genova (CUP D33C23001110001) and it is fostered by 
the National Group of Mathematical Physics (GNFM-INdAM (IT)). 
A.G-S.\ is supported by Royal Society Enhancement Grants (RF\textbackslash ERE\textbackslash 210053
and RF\textbackslash ERE\textbackslash 231077).
G.M.\ is supported by a PhD scholarship of the University of Genova (IT).
A.S.\ gratefully acknowledges the support of 
the Royal Society (UK) through a Royal Society University 
Research Fellowship (URF\textbackslash R\textbackslash 211015)
and Enhancement Grants (RF\textbackslash ERE\textbackslash 210053 and 
RF\textbackslash ERE\textbackslash 231077).

%
%
%





\begin{thebibliography}{10}

\bibitem[B\"ar15]{Bar}
C.~B\"ar,
``Green-hyperbolic operators on globally hyperbolic spacetimes,'' 
Commun.\ Math.\ Phys.\ {\bf 333}, no.\ 3, 1585 (2015) 
[arXiv:1310.0738 [math-ph]].


\bibitem[BF09]{Ginoux}
C.~B\"ar and K.~Fredenhagen (eds.),
{\it Quantum field theory on curved spacetimes: Concepts and mathematical foundations}, 
Lecture Notes in Physics {\bf 786}, Springer, Berlin (2009).


\bibitem[BGP07]{BGP}
C.~B\"ar, N.~Ginoux and F.~Pf\"affle, 
{\it Wave equations on Lorentzian manifolds and quantization}, 
ESI Lect.\ Math.\ Phys., 
European Mathematical Society (EMS), Z\"urich (2007)
[arXiv:0806.1036 [math.DG]].


\bibitem[BMS23]{Benini}
M.~Benini, G.~Musante and A.~Schenkel,
``Green hyperbolic complexes on Lorentzian manifolds,''
Commun.\ Math.\ Phys.\ \textbf{403}, no.\ 2, 699--744 (2023)
[arXiv:2207.04069 [math-ph]].


\bibitem[BSV24]{Bruinsma}
S.~Bruinsma, A.~Schenkel and B.~Vicedo,
``Universal first-order Massey product of a prefactorization algebra,''
Commun.\ Math.\ Phys.\ \textbf{405}, no.\ 9, 206 (2024)
[arXiv:2307.04856 [math-ph]].


\bibitem[CS]{condensed}
D.~Clausen and P.~Scholze,
\textit{Condensed mathematics and complex geometry},
\url{https://people.mpim-bonn.mpg.de/scholze/Complex.pdf}.


\bibitem[CG17]{CG1}
K.~Costello and O.~Gwilliam,
{\it Factorization algebras in quantum field theory: Volume 1},
New Mathematical Monographs {\bf 31}, 
Cambridge University Press, Cambridge (2017).


\bibitem[CG21]{CG2}
K.~Costello and O.~Gwilliam,
{\it Factorization algebras in quantum field theory: Volume 2},
New Mathematical Monographs {\bf 41}, 
Cambridge University Press, Cambridge (2021).


\bibitem[Cra04]{HPT}
M.~Crainic,
``On the perturbation lemma, and deformations,'' 
arXiv:math.AT/0403266.


\bibitem[DF03]{Duetsch}
M.~Duetsch and K.~Fredenhagen,
``The Master Ward Identity and generalized Schwinger-Dyson equation in classical field theory,''
Commun.\ Math.\ Phys.\ \textbf{243}, 275--314 (2003)
[arXiv:hep-th/0211242].


\bibitem[FR12]{FR1}
K.~Fredenhagen and K.~Rejzner,
``Batalin-Vilkovisky formalism in the functional approach to classical field theory,''
Commun.\ Math.\ Phys.\ \textbf{314}, 93--127 (2012)
[arXiv:1101.5112 [math-ph]].


\bibitem[FR13]{FR2}
K.~Fredenhagen and K.~Rejzner,
``Batalin-Vilkovisky formalism in perturbative algebraic quantum field theory,''
Commun.\ Math.\ Phys.\ \textbf{317}, 697--725 (2013)
[arXiv:1110.5232 [math-ph]].


\bibitem[Fr\"o19]{Frob}
M.~B.~Fr\"ob,
``Anomalies in time-ordered products and applications to the BV\textendash{}BRST 
formulation of quantum gauge theories,''
Commun.\ Math.\ Phys.\ \textbf{372}, no.\ 1, 281--341 (2019)
[arXiv:1803.10235 [math-ph]].


\bibitem[Gwi12]{Gwilliam}
O.~Gwilliam, 
{\it Factorization algebras and free field theories}, 
PhD thesis, Northwestern University (2012). 
\url{https://people.math.umass.edu/~gwilliam/thesis.pdf}.


\bibitem[HR20]{Hawkins}
E.~Hawkins and K.~Rejzner,
``The star product in interacting quantum field theory,''
Lett.\ Math.\ Phys.\ \textbf{110}, no.\ 6, 1257--1313 (2020)
[arXiv:1612.09157 [math-ph]].


\bibitem[Hol08]{Hollands}
S.~Hollands,
``Renormalized quantum Yang-Mills fields in curved spacetime,''
Rev.\ Math.\ Phys.\ \textbf{20}, 1033--1172 (2008)
[arXiv:0705.3340 [gr-qc]].


\bibitem[Jar81]{Jarchow}
H.~Jarchow,
{\it Locally convex spaces},
Math.\ Leitf\"aden, 
B.\ G.\ Teubner, Stuttgart (1981).


\bibitem[JRSW19]{Jurco}
B.~Jur\v{c}o, L.~Raspollini, C.~S\"amann and M.~Wolf,
``$L_\infty$-algebras of classical field theories and the Batalin-Vilkovisky formalism,''
Fortsch.\ Phys.\ \textbf{67}, no.\ 7, 1900025 (2019)
[arXiv:1809.09899 [hep-th]].


\bibitem[Kel24]{Kelly}
J.~Kelly, 
{\it Homotopy in exact categories},
Mem.\ Amer.\ Math.\ Soc.\ {\bf 298}, no.\ 1490 (2024)
[arXiv:1603.06557 [math.CT]].


\bibitem[Kna02]{Knapp2002}
A.~W.~Knapp,
\textit{Lie groups beyond an introduction},
Progress in Mathematics \textbf{140},
Birkh\"auser, Boston (2002).


\bibitem[K\"ot69]{Koethe}
G.~K\"othe,
{\it Topological vector spaces I},
Die Grundlehren der mathematischen Wissenschaften {\bf 159},
Springer, New York (1969).


\bibitem[KS24]{KraftSchnitzer}
A.~Kraft and J.~Schnitzer,
``An introduction to $L_\infty$-algebras and their homotopy theory for the working mathematician,''
Rev.\ Math.\ Phys.\ \textbf{36}, no.\ 01, 2330006 (2024)
[arXiv:2207.01861 [math.QA]].


\bibitem[LurX]{Lurie}
J.~Lurie,
{\it Derived algebraic geometry X: Formal moduli problems}.
\url{http://www.math.harvard.edu/~lurie/papers/DAG-X.pdf}.


\bibitem[Pri10]{Pridham}
J.~P.~Pridham,
``Unifying derived deformation theories,''
Adv.\ Math.\ {\bf 224}, 772--826 (2010)
[arXiv:0705.0344 [math.AG]].


\bibitem[Rej16]{RejznerBook}
K.~Rejzner,
{\it Perturbative algebraic quantum field theory: An introduction for mathematicians},
Mathematical Physics Studies, Springer, Cham (2016).


\bibitem[Rej22]{RejznerProceedings}
K.~Rejzner,
``BV quantization in perturbative algebraic QFT: Fundamental concepts and perspectives,''
in: J.~Read and N.~J.~Teh (eds.),
{\it The philosophy and physics of Noether's theorems: A centenary volume},
Cambridge University Press, Cambridge (2022)
[arXiv:2004.14272 [math-ph]].


\bibitem[Rud91]{Rudin}
W.~Rudin, 
{\it Functional analysis},
McGraw-Hill, New York (1991). 


\bibitem[TT19]{TaslimiTehrani}
M.~Taslimi Tehrani,
``Quantum BRST charge in gauge theories in curved space-time,''
J.\ Math.\ Phys.\ \textbf{60}, no.\ 1, 012304 (2019)
[arXiv:1703.04148 [hep-th]].


\bibitem[TTZ20]{Zahn}
M.~Taslimi Tehrani and J.~Zahn,
``Background independence in gauge theories,''
Annales Henri Poincar{\'e} \textbf{21}, no.\ 4, 1135--1190 (2020)
[arXiv:1804.07640 [math-ph]].


\bibitem[War83]{Warner1983}
F.~W.~Warner,
{\it Foundations of differentiable manifolds and Lie groups},
Graduate Texts in Mathematics {\bf 94}, 
Springer, Berlin Heidelberg (1983).


\bibitem[Wei94]{Weibel}
C.~A.~Weibel, 
{\it An introduction to homological algebra}, 
Cambridge University Press, Cambridge (1994).


\bibitem[Wel07]{Wells2007}
R.~O.~Wells,
{\it Differential analysis on complex manifolds},
Graduate Texts in Mathematics {\bf 65}, 
Springer, New York (2007).



\end{thebibliography}
\end{document}